%% file: main.tex
\documentclass{article}
\usepackage[left=2.5cm,top=3cm,right=2.5cm,bottom=3cm,bindingoffset=0.5cm]{geometry}
\usepackage{graphicx} % Required for inserting images   
\usepackage{xcolor}
\usepackage{amsmath}
\usepackage{amsfonts}
\usepackage{amsthm}
\usepackage{amsfonts}
\usepackage{tikz}
\usepackage{hyperref}
\usetikzlibrary{patterns, shapes, knots}
\usetikzlibrary{positioning,arrows}
\usetikzlibrary{decorations.markings}
\usetikzlibrary{decorations.pathmorphing}
\usetikzlibrary{arrows.meta}
\usetikzlibrary{calc,math}
\tikzstyle{every node}=[font=\small]
\usepackage{subfig}

\newtheorem{thm}{Theorem}
\newtheorem{lem}[thm]{Lemma}
\newtheorem{cor}[thm]{Corollary}

\theoremstyle{definition}
\newtheorem{defn}[thm]{Definition}

\newtheorem{question}[thm]{Question}

\numberwithin{thm}{section}

\newcommand{\initcube}{
\tikzmath{
\posx =0; \posy = 0;
\x1 = 1; \y1 = 0; \z1 = 0; 
\x2 = 0; \y2 = 1; \z2 = 0; 
\x3 = 0; \y3 = 0; \z3 = 1;}
}

\newcommand{\rolldown}{
\tikzmath{
\posy = \posy - 1;
\t = \y1; \y1 = -\z1; \z1 = \t; 
\t = \y2; \y2 = -\z2; \z2 = \t; 
\t = \y3; \y3 = -\z3; \z3 = \t;}
}
\newcommand{\rollup}{
\tikzmath{\posy = \posy + 4;}\rolldown \rolldown \rolldown}

\newcommand{\flipdown}{
\tikzmath{
\posy = \posy - 1;
\t = \y1; \y1 = -\t; 
\t = \y2; \y2 = -\t; 
\t = \y3; \y3 = -\t;}
}
\newcommand{\flipup}{
\tikzmath{\posy = \posy + 2;} \flipdown} 

\newcommand{\rollright}{
\tikzmath{
\posx = \posx + 1;
\t = \x1; \x1 = -\z1; \z1 = \t; 
\t = \x2; \x2 = -\z2; \z2 = \t; 
\t = \x3; \x3 = -\z3; \z3 = \t; }
}
\newcommand{\rollleft}{\tikzmath{\posx = \posx -4;}\rollright \rollright \rollright}

\newcommand{\flipright}{
\tikzmath{
\posx = \posx + 1;
\t = \x1; \x1 = -\t; 
\t = \x2; \x2 = -\t; 
\t = \x3; \x3 = -\t;}
}
\newcommand{\flipleft}{\tikzmath{\posx = \posx - 2;}\flipright}

\newcommand{\printcube}{

\pgfmathtruncatemacro\a{\z3!=0?\x1:(\z1!=0?\x2:\x3)}
\pgfmathtruncatemacro\c{\z3!=0?\y1:(\z1!=0?\y2:\y3)}
\pgfmathtruncatemacro\b{\z3!=0?\x2:(\z1!=0?\x3:\x1)}
\pgfmathtruncatemacro\d{\z3!=0?\y2:(\z1!=0?\y3:\y1)}

\pgfmathtruncatemacro\rot{\a==0?90:0}
\pgfmathtruncatemacro\xsc{\a+\b}
\pgfmathtruncatemacro\ysc{\d-\c}

\pgfmathtruncatemacro\lett{\z1==1?3:(\z2==1?2:(\z3==1?1:(\z1==-1?4: (\z2==-1?5:(\z3==-1?6:0)))))}
\pgfmathtruncatemacro\presc{\z1+\z2+\z3}
\pgfmathtruncatemacro\prerot{\z1==1?90:(\z3==0?-90:0)}
\pgfmathtruncatemacro\col{80+20*\xsc*\ysc*\presc}
  
\path[shift = {(\posx,\posy)},fill = gray!10] (-.5,-.5)--(-.5,.5)--(.5,.5)--(.5,-.5)--cycle;
\node[text height=.5cm,text width=.5cm,color = black!\col,xscale = \xsc, yscale = \ysc,rotate=\rot,xscale = \presc, rotate = \prerot] at (\posx,\posy) {\bfseries \Large \lett};
}

\title{Bounding the number of holes required for folding rectangular polyominoes into cubes}
\author{Florian Lehner and Benjamin Shirley}
\date{January 2025}

\begin{document}
\maketitle
\begin{abstract}
We study the problem of whether rectangular polyominoes with holes are cube-foldable, that is, whether they can be folded into a cube, if creases are only allowed along grid lines. It is known that holes of sufficient size guarantee that this is the case. 

Smaller holes which by themselves do not make a rectangular polyomino cube-foldable can sometimes be combined to create cube-foldable polyominoes. We investigate minimal sets of holes which guarantee cube-foldability. We show that if all holes are of the same type, the these minimal sets have size at most 4, and if we allow different types of holes, then there is no upper bound on the size. 
\end{abstract}

\section{Introduction}

Origami and paper folding present a fruitful source of mathematical problems, and the subject as a whole has seen extensive study,
\cite{geometricfolding,origametry}.
In light of this it is perhaps surprising how little is known about many seemingly simple problems including the problem of folding polyominoes into polycubes.

A polyomino is a two-dimensional polyhedron, formed by taking some subset of the faces of the square lattice in the plane.
Although we require polyominoes to be connected, we allow cuts between squares to create holes in various shapes, in particular we do not require that all pairs of squares which share an edge in the grid are connected via that edge.
A polycube is the three-dimensional analogue of the polyomino, being a connected three-dimensional polyhedron, formed from a union of unit cubes. Given a polyomino, $P$, we would like to decide whether $P$ can be folded in a way to cover the surface of some given polycube.

Motivated by a recreational problem \cite{puzzle}, problems of this form appear to have first been examined in \cite{original}, where various folding models were considered.
This work was then expanded on in \cite{aichholzer2020}, which in particular focused on polyominoes containing holes.
Following on from \cite{aichholzer2020} and \cite{florian}, in this paper we further investigate the problem of whether a given polyomino can be folded into a unit cube in the perhaps simplest folding model,
where folds are only allowed along edges of the polyomino. Let us call a polyomino \emph{cube-foldable}, or simply \emph{foldable}, if this is possible.

It was conjectured in \cite{florian} that there is a polynomial time algorithm deciding the foldability of polyominoes.
In the same paper, the first provably correct general algorithm for this decision problem was presented, 
but this algorithm has exponential run time and needs to solve the unlink recognition problem from topology as a subroutine.
An implementation of this algorithm can be found in the second authors Github repository \cite{ben-github}.
More tractable subproblems are known to have polynomial time solutions. 
For example, a full classification of foldable tree-shaped polyominoes exists, see \cite{florian}, and thus we can recognize foldability of such polyominoes in polynomial time.

Much of the complexity of this problem appears to lie in the ways in which holes in polyominoes affect foldability, and how different holes 
can interact with each other in various foldings. We refer the reader to \cite{puzzle},
and encourage them to try these puzzles for themselves. 

In the present paper, we focus on the foldability of rectangular polyominoes with holes, 
making progress on Open Problem 8.1 from \cite{florian}. It was shown in \cite{aichholzer2020} that if a polyomino has a hole other than five so-called simple holes, 
then it can always be folded into a unit cube. Four of these simple holes are shown in Figure \ref{fig:4holes}. 
The last one, a slit of length 1, likely does not change foldability since any (valid) way of folding of the section of paper surrounding this hole folds as though the hole was non-existent,
see \cite{aichholzer2020}, Section 4.2. Although adding any one of the simple holes is not enough to make a rectangular polyomino foldable,
adding two or more of them can produce a foldable polyomino.

\begin{figure}
    \centering
    \subfloat[unit square]{
        \begin{tikzpicture}
            \fill[gray!10] (0.1, 0.1) rectangle (2.9, 2.9);
            \draw[gray, very thin] (0.1, 0.1) grid (2.9, 2.9);
            \draw[blue, very thick] (1, 1) rectangle (2, 2);
        \end{tikzpicture}
    }
    \subfloat[L-shape]{
        \begin{tikzpicture}
            \fill[gray!10] (0.1, 0.1) rectangle (2.9, 2.9);
            \draw[gray, very thin] (0.1, 0.1) grid (2.9, 2.9);
            \draw[blue, very thick] (1, 1) -- (2, 1);
            \draw[blue, very thick] (1, 1) -- (1, 2);
        \end{tikzpicture}
    }
    \subfloat[U-shape]{
        \begin{tikzpicture}
            \fill[gray!10] (0.1, 0.1) rectangle (2.9, 2.9);
            \draw[gray, very thin] (0.1, 0.1) grid (2.9, 2.9);
            \draw[blue, very thick] (1, 1) -- (2, 1);
            \draw[blue, very thick] (1, 1) -- (1, 2);
            \draw[blue, very thick] (2, 1) -- (2, 2);
        \end{tikzpicture}
    }
    \subfloat[2 tall slit]{
        \begin{tikzpicture}
            \fill[gray!10] (0.1, 0.1) rectangle (1.9, 3.9);
            \draw[gray, very thin] (0.1, 0.1) grid (1.9, 3.9);
            \draw[blue, very thick] (1, 1) -- (1, 3);
        \end{tikzpicture}
    }
    \caption{The $4$ simple holes.}
    \label{fig:4holes}
\end{figure}
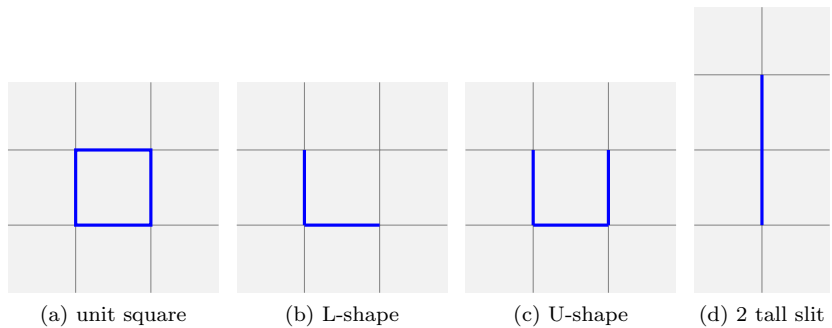

We investigate how large a minimal set of simple holes which guarantees foldability can be. 
More precisely, let us say that a set of holes \emph{cooperate}, if a rectangular polyomino with these holes is foldable, 
and that such a set \emph{minimally cooperates} if it cooperates, but no strict subset does. 
For one of these simple holes, the unit square hole, it was shown in \cite{florian} that a set of unit square holes cooperates if and only if some subset of two unit square holes cooperates. 
In other words, all minimally cooperating sets have size $2$ (and a full classification is known). This raises the question of whether something similar holds for the other simple holes.

\begin{question}
    Is there a constant, $c$, such that any cooperating set of simple holes contains a cooperating subset of size $c$?
\end{question}

We show that this question has a positive answer if either all or none of the simple holes in the set are slits of length $2$, and a negative answer if we mix slits of length $2$ with other types of simple holes.

More precisely, in section \ref{sec:3holesbound} we show that any minimally cooperating set of square, L-shaped and U-shaped holes has size $2$.
In section \ref{sec:slitBound}, we investigate minimally cooperating sets of slits. For polyominoes of size at least $6\times 6$, any such minimally cooperating set has again size 2, but perhaps surprisingly, there are minimally cooperating sets of size $4$ in the $4\times n$ case, and of size $3$ in the $5\times n$ case. 
Complete characterisations of minimally cooperating sets of holes can be extracted from these results.
Finally in section \ref{sec:unbounded}, we demonstrate that there are polyominoes containing both square and 2-tall slit holes with arbitrarily large minimally cooperating sets of holes.

\section{Definitions and Notation}
\label{sec:defs}
\subsection{The folding model}

A polyomino is a finite collection of faces of the square grid in $\mathbb R^2$, where boundary edges of adjacent faces can, but do not have to be identified. As mentioned in the introduction, in this paper we study when polyominoes fold can be folded into a unit cube, where folds are only allowed along grid lines, and only at angles of $90^\circ$ and $180^\circ$. Intuitively, folding a polyomino into a cube amounts to ``deforming'' the polyomino in $\mathbb R^3$, such that each face of the polyomino ends up covering a face of the cube. Making this intuition (as well as the definition of a polyomino) rigorous requires some topology. 

Following the exposition in \cite{florian}, an \emph{abstract polyomino}, $P$, is a square complex in which every edge and vertex are contained inside at least one face, and each edge is contained in at most two faces.
A \emph{folding blueprint} is a mapping, $\beta: P\to\mathbb{R}^3$, in which the restriction of $\beta$ to each face of $P$ is an isometry. Now we can define a \emph{polyomino} as an abstract polyomino that admits a flat folding blueprint, that is, a map $P\to \mathbb R^2$ such that the images of different squares overlap only at edges. Note that if a polyomino has a flat folding blueprint, then this flat folding blueprint is unique up to isometries of $\mathbb R^2$.

A \emph{folded state}, $F$, of a polyomino, $P$, is a piecewise linear embedding, $F:P\to\mathbb{R}^3$.
$F$ is said to be a $\varepsilon$-realization (or simply, a realization) of a folding blueprint, $\beta$, if $||F(x)-\beta(x)||<\varepsilon$ for every $x\in P$.

We say two folded states are equivalent if they are ambient isotopic. Further, if a folded state, $F$, is equivalent to 
a realization of a flat folding blueprint, then we call $F$ a \emph{valid folding}, or sometimes just a \emph{folding}. Finally, taking the surface of the unit cube, 
$\mathcal{C}\subseteq \mathbb{R}^3$, we say $F$ is a folding into the cube if $F$ corresponds to some folding blueprint $\beta: P\to\mathcal{C}$. Additionally, if $F$ covers all faces of the cube, we say it is a folding \emph{onto} the cube. As mentioned earlier, we call a polyomino \emph{cube-foldable}, or sometimes just \emph{foldable} if it admits a folding onto the cube.

Clearly, any folding, $F$, of a polyomino onto the cube uniquely determines the following two (combinatorial) objects. 

A \emph{facemapping} of a polyomino into the unit cube, is a homomorphism, $\varphi$, from the faces, vertices, and edges of the polyomino to the faces, vertices and edges of the cube, which preserves the incidence relation. 
The facemapping corresponding to a folding, $F$, is given by the image of each face under the corresponding folding blueprint.
We will often identify the faces of the cube with the integers $1$ through $6$ (see Figure \ref{fig:basicFacemap}), and describe our mapping in terms of these integers; throughout the paper, the labelling of the faces is assumed to be the labelling obtained from folding the cube net in Figure \ref{fig:basicFacemap}.
We frequently describe a facemapping by overlaying these integers (with various applied rotations and reflections) to the squares of a polyomino (see Figure \ref{fig:squareHoleFolding} for an example).

\begin{figure}
    \centering
    \begin{tikzpicture}[xscale=0.7,yscale=0.7]
        \initcube\printcube
        \rollup\printcube
        \rolldown\rollright\printcube
        \rollleft\rolldown\printcube
        \rollup\rollleft\printcube\rollleft\printcube
        \begin{scope}[shift={(-.5,-.5)}] 
            \draw[blue, very thick] (1, 1) -- (1, 2);
            \draw[blue, very thick] (1, 2) -- (0, 2);
            \draw[blue, very thick] (0, 2) -- (0, 1);
            \draw[blue, very thick] (0, 1) -- (-2, 1);
            \draw[blue, very thick] (-2, 1) -- (-2, 0);
            \draw[blue, very thick] (-2, 0) -- (0, 0);
            \draw[blue, very thick] (0, 0) -- (0, -1);
            \draw[blue, very thick] (0, -1) -- (1, -1);
            \draw[blue, very thick] (1, -1) -- (1, 0);
            \draw[blue, very thick] (1, 0) -- (2, 0);
            \draw[blue, very thick] (2, 0) -- (2, 1);
            \draw[blue, very thick] (2, 1) -- (1, 1);
            \draw[gray] (1, 1) -- (0, 1);
            \draw[gray] (0, 1) -- (0, 0);
            \draw[gray] (-1, 1) -- (-1, 0);
            \draw[gray] (0, 0) -- (1, 0);
            \draw[gray] (1, 0) -- (1, 1);
        \end{scope}
    \end{tikzpicture}
    \caption{Standard Cube Layout}
    \label{fig:basicFacemap}
\end{figure}
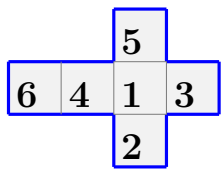

A \emph{layer mapping}, $l$, is a mapping that assigns to each face of the polyomino, an integer between $1$ and the number of squares mapping to the same face on the cube, such that restricting $l$ to the squares that map to a particular face is a bijection. The layermapping corresponding to a folding,
$F$, is given by the order (inside to outside of the unit cube) in which the faces of $P$ are mapped in the neighbourhood of a face of $\mathcal C$.
In our figures, we sometimes use small numbers in the corners of faces to indicate the layer mapping.

\begin{defn}
    We call a triple $(P, \varphi, l)$ a \emph{valid pseudo-folding} if there is some valid folding, $F:P\to\mathcal{C}$, which defines $\varphi$ and $l$.
\end{defn}

It is clear that the existence of some onto facemapping, $\varphi$, is a necessary condition for a polyomino, $P$, to fold onto $\mathcal{C}$.
Therefore, one way to disprove the foldability of $P$ is to prove that no such facemapping can exist.
Conversely, in order to prove that a polyomino does admit a valid folding, we take the pragmatic approach of providing a pseudo-folding whose validity can be checked by folding (a paper model of) the polyomino.

Note that face mappings make no distinction between $+180^\circ$ and $-180^\circ$ folds.
A $180^\circ$ fold between two squares implies they map to the same face, regardless of the direction of this fold.
Moreover, it is clear that if a square is folded to a face under a $+90^\circ$ fold, then folding this crease 
$-90^\circ$ will not allow it to map to a face anymore. Therefore, when discussing facemappings, we make no distinction between the 
directions creases are folded, and thus we simply describe all folds as either $180^\circ$ or $90^\circ$ folds.

\subsection{Further Definitions and Basic Lemmas}

Consider a flat folding blueprint, $\beta$, of a polyomino, $P$. Each connected component of the boundary of $P$ is homeomorphic to a circle, and thus maps to a closed curve (which may intersect itself non-trivially). Since $P$ is bounded and connected, $\beta(P)$ lies entirely on one side (inside or outside) of each such closed curve, and there is exactly one such closed curve for which this is the bounded side. We call this closed curve the \emph{shape} of $P$ and all other closed curves \emph{holes}. By slight abuse of notation, we also call the preimage of a hole under $\beta$ a hole, that is, we view the hole as a subset of the polyomino.

We will mostly concern ourselves with rectangular polyominoes, that is, polyominoes whose shape is a rectangle. In \cite{aichholzer2020}, the following lemma regarding such polyominoes was shown. We use this lemma extensively in our results.
\begin{lem}
    \label{lem:rectangularSection}
    Let $k,n\ge 2$ and let $P$ be a rectangular polyomino without holes.
    Then, in every folding of $P$ into $\mathcal C$, any pair of collinear creases are either both folded
    by $90^\circ$ or both folded by $180^\circ$. 
    Moreover, either all horizontal or all vertical creases of $P$ are folded by $180^\circ$.
\end{lem}

Note that we may apply this lemma to sub-polyominoes, $P'$, of a larger polyomino, $P$. This ensures that every rectangular, hole-free sub-polyomino of any polyomino must have a complete set of $180^\circ$ folds in one direction in any folding of $P$ into $\mathcal C$.
Holes sometimes allow us to fold rectangular polyominoes in ways we could not fold a regular rectangular section of paper. In particular, if a polyomino contains a hole other than the five simple holes shown in Figure \ref{fig:4holes} then it has a folding onto the cube \cite{aichholzer2020}. A combination of simple holes can also guarantee cube-foldability, but we do not always need all such holes.

\begin{defn}
    \label{def:trivialhole}
   Let $h$ be a simple hole in a polyomino, $P$, and let $(P, \varphi, l)$ be a valid pseudo-folding. 
   We say that $h$ folds \emph{trivially} if any two sides of $h$ have the same image under a flat folding blueprint $\beta$ if and only if they have the same image under $\varphi$. Otherwise we say $h$ is folded \emph{non-trivially}.
\end{defn}
Note that if $h$ is a unit square hole, then $h$ is folded trivially if and only if its 4 sides map to the 4 sides of some face of the cube. In particular, a hole $h$ is folded trivially if we could "fill in" $h$, or pretend it does not exist, and still fold $P$ in the same way as before.

Let $S$ be a set of holes of a polyomino, $P$. As mentioned in the introduction, we say that $S$ \emph{cooperates} if there is a folding of $P$ such that every hole, $h \notin S$, is folded trivially. Moreover, $S$ \emph{minimally cooperates} if $S$ cooperates but no subset of $S$ does. Note that a polyomino may have multiple minimally cooperating sets, which naturally implies that this polyomino would have multiple foldings.

\begin{defn}
    \label{def:support}
The \textit{support} of a set of holes, $\{h_1, ..., h_n\}$, in a polyomino, $P$, is the smallest rectangular sub-polyomino that contains all holes in the set.
\end{defn}

\section{A bound for Square, L-shaped and U-shaped Holes}
\label{sec:3holesbound}

Observing the non-trivial foldings of the L-shaped and U-shaped holes, we note that those holes fold in a very similar way to the unit square hole. It was shown in \cite[Lemma 10]{aichholzer2020} that up to rotation and reflection, 
there is only one non-trivial facemapping for a unit square hole, seen in Figure \ref{fig:squareHoleFolding}.
In \cite[Theorem 6.2]{florian}, it was shown that a rectangular polyomino containing only unit square holes
is foldable if and only if some pair of holes cooperate, hence showing the largest minimally cooperating set is of size $2$.
We generalise this proof to include L-shaped and U-shaped holes.

\begin{figure}
    \centering
    \begin{tikzpicture}[xscale=0.7,yscale=0.7]
        \initcube\printcube
        \rollup\printcube
        \flipup\printcube
        \rollright\printcube
        \rollright\printcube
        \rolldown\printcube
        \flipdown\printcube
        \rollleft\printcube
        \begin{scope}[shift={(-.5,-.5)}] 
            \draw[gray] (0.1, 0.1) grid (2.9, 2.9);
            \fill[white] (1, 1) rectangle (2, 2);
            \draw[blue, very thick] (1, 1) rectangle (2, 2);
        \end{scope}
    \end{tikzpicture}
    \caption{Facemapping of a non-trivially folded square hole}
    \label{fig:squareHoleFolding}
\end{figure}

\begin{thm} 
    \label{thm:3holebound}
    A rectangular polyomino containing only unit square, L-shaped and U-shaped holes is foldable if and only if some two of its holes cooperate.
\end{thm}
\begin{proof}
    The forward direction for this proof is trivial. Therefore, it remains to prove that if a polyomino, $P$, satisfying the conditions of the theorem folds, then $P$ contains two holes that cooperate.

    Much like the proof in \cite{florian}, we proceed by induction on the number of holes, but we need some preliminary observations.

    Firstly, observe the following properties of the L and U-shaped holes.

    U-shaped hole have (up to rotation and reflection) two different types of non-trivial foldings.
    Ignoring the central flap for a moment, notice that we can fold the remaining $8$ squares that 
    define the hole either trivially or non-trivially.
    
    If they fold non-trivially, we can see that the eight squares surrounding the hole 
    are folded exactly like a unit square hole, covering $5$ sides of the cube.
    The middle flap then must cover one of these $5$ faces, as all $4$ central edges map 
    to the face of the cube opposite the face left uncovered by the folding. In particular, the flap plays no role for the faces of the cube covered by the folded polyomino, and hence we do not distinguish the two possible ways of folding it.
    When a U-shaped hole folds like this, we say it is \emph{folded like a unit square hole}. 

    The other way these $8$ squares can fold is trivially, in exactly the same way as a $3\times 3$ rectangle would. 
    Because the flap is only attached by a single edge, it can map to a different face than the rest of the hole, creating a 
    ``T-shape''; in this case we say that the U-shaped hole is \emph{T-folded}.

    In any consistent facemapping, $\varphi$, of a polyomino where a U-shaped hole is T-folded,
    note that the square with only a single edge connected to the polyomino essentially works independently:
    the face this square maps to does not affect the consistency of the mapping, so long as it remains incident to this edge.
    Therefore, whenever such a hole is folded like this, we can define a new mapping, $\varphi'$, by mapping this face to 
    wherever it would go if it were folded trivially. Clearly this respects the incidence relation and so this is a consistent mapping too.
    It is easy to see that the layer mapping can be adjusted to preserve the validity of this folding,
    simply by putting the flap on any layer that will not cause a self-intersection. Note that the resulting folding is always a folding into the cube, but may or may not be a folding onto the cube depending on whether the face of the cube covered by the flap is also covered by other faces of the polyomino.

    For the L-shaped hole, it was shown in \cite{aichholzer2020} that the only non-trivial foldings are those that cover $5$ faces of the cube, similar to the unit square hole, but more limited due to the asymmetry. 
    
    We can thus essentially view an L-shaped hole as a restricted version of the square hole, except that it can appear ``closer'' to other holes than the square hole in the sense that the support of an L-shaped hole may contain another simple hole, see Figure \ref{fig:LshapeSquareOverlap}.
    Examining all possible facemappings of the polyomino in Figure \ref{fig:LshapeSquareOverlap} where the L-shaped is folded non-trivially, we note that the unit square hole is folded trivially in all of them. Thus, if an L-shaped hole and a square hole appear in the configuration shown in Figure \ref{fig:LshapeSquareOverlap}, then  in any folding onto the unit cube, one of these holes must fold trivially and therefore the set of holes does not minimally cooperate. An analogous argument applies if the square hole is replaced by another L-shaped hole.

    \begin{figure}
        \centering
        \input{Lshaped_square}
        \caption{An L-shaped hole, which has overlapping closure with a 
        square hole}
        \label{fig:LshapeSquareOverlap}
    \end{figure}
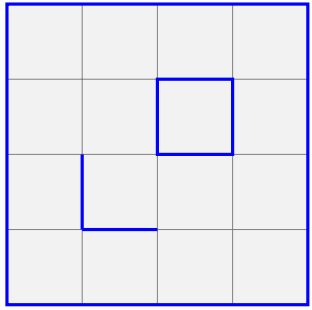
    \begin{figure}
        \centering
        \input{Lshaped_Ushaped}
        \caption{An L-shaped hole, with a U-shaped hole that folds.}
        \label{fig:LshapeUshapeOverlap}
    \end{figure}
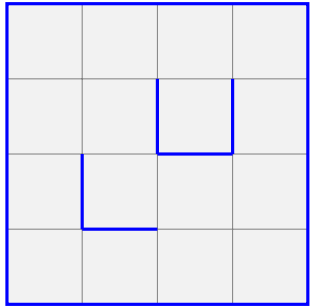
    If the square hole is replaced by a U-shaped hole, then there are orientations of the U-shaped 
    hole which result in a foldable polyomino, see Figure \ref{fig:LshapeUshapeOverlap}. 
    If the two holes cooperate, we are done. If they do not cooperate, and both holes fold non-trivially, then the U-shaped  hole must be T-folded. In this case, the flap of the U-shaped hole will necessarily cover a face which is also covered by some square bordering the L-shaped hole. In particular, if the polyomino in question has a folding onto the cube, then it also has a folding onto the cube where one of the two holes is folded trivially.

    We are now ready to proceed with the proof. As mentioned earlier, we will use induction on the number of holes $k$ to show that any foldable rectangular polyomino with $k$ holes contains a pair of cooperating holes.

    The base case $k=2$ holds trivially, so assume the theorem holds for 
    all polyominoes with some $k\ge 2$ holes.
    Consider a foldable rectangular polyomino $P$, with $k+1$ holes, and let $(\varphi, l)$ be a valid pseudo-folding of $P$ into a cube.
    Then, under $\varphi$ every hole must be folded non-trivially, as otherwise, we would have a cooperating set of $k$ holes and could apply the induction hypothesis. Hence, by the above discussion we may assume that the support of any L-shaped hole contains no other holes.
    
    We separate the proof into two cases: either there is a hole which folds like a unit square hole under the folding $(\varphi, l)$, or there is no such hole.

    \textbf{Case 1:} Some hole folds like a unit square. To begin,
    construct an alternate mapping $\varphi'$ by taking every U-shaped hole 
    that T-folds, and tucking the flap away such that it is folded trivially, leaving us with a folding in which every non-trivially folded hole folds like a unit square hole. We will show that every such hole must fold to cover the same five faces.

    Construct a graph $G$, where the unit-square folding holes are the vertex set, 
    and two holes are connected if their support (recall Definition \ref{def:support}) contains no other holes that fold like unit square holes.

    Clearly this graph is connected. We claim that the support of any two neighbouring holes
    $h_1, h_2$ in this graph must map to the same five faces. If not, then together they cover all six faces of the cube because the set of squares surrounding each of the holes covers $5$ faces.
     But then we can fold everything outside the support trivially with $180^\circ$-folds, showing that $h_1, h_2$ cooperate.

    Hence, every hole which folds like a unit square must cover the same five faces, and thus the final face must be covered by the flap of a T-folded U-shaped hole. If there is more than one such T-folded U-shaped hole, we can fold all but one of them trivially and still obtain a folding of $P$ into a cube, and thus by induction we may assume that there is exactly one T-folded U-shaped hole and all other holes are folded like a unit square.

    There is a rectangular subpolyomino of $P$ containing this T-folded U-shaped hole, and exactly one other hole. The other hole 
    must cover $5$ faces of the cube under $\varphi$, and the flap covers the last one, so this subpolyomino folds into a cube under $(\varphi, l)$. But this means that we can fold everything outside this subpolyomino trivially with $180^\circ$-folds, showing that the two holes cooperate.

    \textbf{Case 2:} No hole is folded like a square hole, thus every hole is a T-folded U-shaped hole.
    In this case begin by constructing the alternate mapping, $\varphi'$, from $\varphi$ by tucking away the flaps of all T-shaped folds as before.
    Now we are essentially left with a big rectangle, and so by Lemma \ref{lem:rectangularSection}, every row or column 
    must be folded $180^\circ$ under $\varphi'$. Therefore this folding covers at most $4$ faces 
    of the cube in a ring shape, leaving at least $2$ opposite faces uncovered. We assume without loss of generality that all the horizontal creases are folded $180^\circ$.
    Observe that after applying the $180^\circ$ folds, every even row of vertices gets mapped to one side of the ring,
    and every odd row of vertices gets mapped to the other side. Thus there must be a T-folded U-shaped hole whose flap is attached at an even row, and one whose flap is attached at an odd row.

    If the polyomino is $4$ or more columns wide, then these two holes cooperate. The presence of any U-shaped hole in $P$ guarantees that the width of the polyomino is at least $3$. If the width is exactly 3, then $\varphi'$ only covers at most $3$ squares. In any orientation of a U-shaped hole in this polyomino, due to the $180^\circ$ folds, at least one vertex of the 
    hole's connected edge touches a square mapping to the center of the $1\times 3$ strip these $180^\circ$ folds create. Thus the face of the cube opposite to the one this center column maps to is left uncovered, thus contradicting the assumption that $\varphi$ was a folding of $P$ onto the cube.
\end{proof}

\section{A Bound for 2-Tall Slits}
\label{sec:slitBound}
The previous proof shows that for three of the four holes we are considering, any minimally cooperating set has size exactly $2$. 
For polyominoes containing only slits things are slightly more complicated.
For instance, the polyomino in Figure \ref{fig:3slitexample} is foldable, but no two of its holes cooperate. 
Before showing any results, we first introduce some basic terminology regarding these slit holes and present a few lemmas.

\begin{figure}[hbt]
    \centering
    \begin{tikzpicture}[xscale=0.7,yscale=0.7]
        \initcube\printcube
        \rollup\printcube
        \rollup\printcube
        \rollup\printcube
        \rollup\printcube
        \flipright\printcube
        \flipdown\printcube
        \rolldown\printcube
        \flipdown\printcube
        \rolldown\printcube
        \rollright\printcube
        \rollup\printcube
        \flipup\printcube
        \rollup\printcube
        \flipup\printcube
        \flipright\printcube
        \rolldown\printcube
        \rolldown\printcube
        \rolldown\printcube
        \rolldown\printcube
        
        \begin{scope}[shift={(-.5,-.5)}]
            \draw[gray, very thin] (0, 0) grid (4, 5);
            \draw[blue, very thick] (0, 0) rectangle (4, 5);
            \draw[blue, very thick] (1, 2) -- (1, 4);
            \draw[blue, very thick] (3, 2) -- (3, 4);
            \draw[blue, very thick] (2, 1) -- (2, 3);
        \end{scope}
    \end{tikzpicture}
    \caption{A polyomino where the removal of any one slit destroys foldability.}
    \label{fig:3slitexample}
\end{figure}

Each non-trival folding of a slit hole falls into one of two distinct classes.
\begin{defn}
Let $h$ be a two-tall slit, and assume without loss of generality that it is vertical, and folded non-trivially. If the central horizontal crease is folded $180^\circ$, we say the slit folds to make a \textit{flap}.
If the central horizontal crease is folded $90^\circ$, we say the slit is folded into a \textit{ring}.
\end{defn}

The following two lemmas show that these are the only non-trivial foldings of these holes,
and also provide some insight into how folds of the surrounding creases interact with slits.

\begin{lem}
    \label{lem:fact1}
    Let $h$ be a vertical slit. The vertical creases directly above and below $h$ are either both folded $180^\circ$ or both folded $90^\circ$ under any folding into the cube. 
    Furthermore, if they are folded $180^\circ$, then the support of $h$ covers at most $4$ faces of the cube, leaving at least two opposite faces of the cube uncovered.
\end{lem} 
\begin{proof}
Consider Figure \ref{fig:fact1} and assume that the crease between squares $a$ and $b$ is folded $180^\circ$;
we refer to the corner in square $x$ labelled with the number $i$ by $x_i$.
Observe that each column of $4$ squares in the support of $h$ maps to at most 4
faces of the cube, and that these faces must all be part of a ring of squares around the cube. 
Further note that $a_1$ and $b_0$ map to the same vertex,
as do $a_2$ and $b_3$. These vertices define the direction of the ring around the cube that each column could map to.
Thus, both columns must map to the same ring of $4$ faces, and all horizontal creases lie in this ring.

The edge between faces $g$ and $h$ is perpendicular to the horizontal creases. Since both $g$ and $h$ lie in the ring, this edge must be folded $180^\circ$.

To see that $90^\circ$ folds propagate as well, assume that the crease between $a$ and $b$ is folded $90^\circ$ and the crease between $g$ and $h$ is folded $180^\circ$. By what we showed above, the $180^\circ$ fold must propagate which is a contradiction.
\begin{figure}[htb]
    \centering
    \begin{tikzpicture}
        \fill[gray!10] (0, 0) rectangle (2, 4);
        \draw[gray, thin] (0, 0) grid (2, 4);
        \draw[blue, thick] (0, 0) rectangle (2, 4);
        \draw[blue, very thick] (1, 1) -- (1, 3);

        \foreach \l [count=\index] in {a,...,h}
            \pgfmathsetmacro\x{Mod(\index,2)==0?1.5:0.5}
            \pgfmathsetmacro\y{div(\index+1,2)-0.5}
            \path (\x, \y) node {\l}
            (\x+0.35, \y+0.25) node {0}
            (\x-0.35, \y+0.25) node {1}
            (\x-0.35, \y-0.25) node {2}
            (\x+0.35, \y-0.25) node {3};
    \end{tikzpicture}
    \caption{$180^\circ$ folds and vertical holes}
    \label{fig:fact1}
\end{figure}
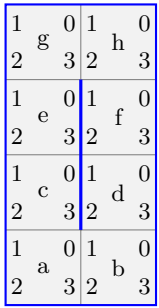

\end{proof}

\begin{lem}
    \label{lem:fact2}
    Let $h$ be a vertical slit. The horizontal creases directly left and right of the centre of $h$ are either both folded $180^\circ$ or both folded $90^\circ$ under any folding into the cube. 
\end{lem} 
\begin{proof}
Again, consider Figure \ref{fig:fact1}. Any $180^\circ$ fold through the crease between $c$ and $e$ forces $f_1$ and $d_2$ to map to the same vertex and thus the edge between $d$ and $f$ must fold $180^\circ$. Just like in Lemma \ref{lem:fact1} we can apply a proof by contradiction to show that $90^\circ$ folds must necessarily propagate too.
\end{proof}

\begin{lem}
\label{lem:fact3}
If a vertical slit is folded $90^\circ$ vertically, it cannot have a $90^\circ$ fold through the center of the hole horizontally.
\end{lem}

\begin{proof}
To see why this is true, assume towards a contradiction that the horizontal creases next to the centre of the slit are folded $90^\circ$. 

\begin{figure}
    \centering
        \begin{tikzpicture}
            \fill[gray!10] (0, 0) rectangle (1, 1);
            \draw[blue, very thick] (0, 0) rectangle (1, 1);
            \draw[red, very thick] (0, 1.2) -- (1, 1.2);
            \draw[red, very thick] (-0.2, 0) -- (-0.2, 1);

            \path (0.2, 0.8) node {$\alpha$}
                (0.8, 0.8) node {$\beta$}
                (0.8, 0.2) node {$\gamma$}
                (0.2, 0.2) node {$\delta$}
                (0.5, 1.4) node {c}
                (-0.4, 0.5) node {d}
                (-0.4, 1) node {*}
                (0, 1.4) node {*};
        \end{tikzpicture}
    \caption{Top-down view of the cube in Lemma \ref{lem:fact3}.}
    \label{fig:fact3}
\end{figure}
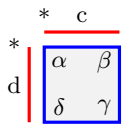
Consider squares $c$ and $e$ in Figure \ref{fig:fact1}. 
The presence of the $90^\circ$ folds in between them necessarily causes $c_3$ and $e_0$ to be mapped to opposite sides of some face of the cube, that is, vertices $\alpha$ and $\gamma$ in Figure \ref{fig:fact3}.

Also notice that if $c_0$ and $d_1$ are mapped to the same vertex on the cube,
then the slit would fold trivially, and by Lemma \ref{lem:rectangularSection},
that makes the two $90^\circ$ folds impossible. Thus, the two vertices must map to different vertices on the cube, and therefore 
must map to vertices $\beta$ and $\delta$ in Figure \ref{fig:fact3}.

If square $c$ mapped to the face of the cube shown in Figure \ref{fig:fact3},
then $c_1$ would map to vertex $\gamma$, but this would mean that both $e_0$ and $e_2$ are mapped to $\gamma$ which is impossible. A similar argument works for square $d$,
and so both squares must map to the faces on the sides of the cube as indicated in Figure \ref{fig:fact3}.

But now, this forces vertices $c_2$ and $d_3$ to be mapped to the same vertex of the cube, and thus the vertical fold between $c$ and $d$ must be folded $180^\circ$.
\end{proof}

Before moving on we need another definition.
\begin{defn}
Call two vertical slits \emph{even-separated}, if there is an even number of rows between their respective midpoints, and \emph{odd-separated} otherwise. Similarly, call two horizontal slits \emph{even-separated} or \emph{odd-separated}, depending on whether there is an even or odd number of columns between their respective midpoints.
\end{defn}

\begin{thm}
    \label{thm:slitsnonfolding}
    Let $P$ be a rectangular polyomino
    containing only slits of length two such that any two horizontal or vertical slits are even-separated. Then $P$ is not cube-foldable.
\end{thm}
\begin{proof}
    We will prove this by showing that for any polyomino matching the above description, there is no consistent facemapping which is onto. 

    Assume towards a contradiction that there is some consistent facemapping, $\varphi$, for a polyomino, $P$, containing only even-separated slits. 
    
    Note that we may assume without loss of generality that there is at least one horizontal and at least one vertical slit: 
    if not, add two rows and two columns with a horizontal or vertical slit,
    respectively and note that if $P$ has a surjective, consistent facemapping, then so does the resulting polyomino. 
    Call a horizontal crease \emph{even} if there is an even number of rows between that crease and the midpoint of any (equivalently: every) vertical slit and \emph{odd} otherwise. 
    Call a vertical crease \emph{even} if there is an even number of columns between that crease and the midpoint of any (equivalently: every) horizontal slit and \emph{odd} otherwise. 
    
    We first show that either every odd horizontal, or every odd vertical 
    crease must be folded $180^\circ$. Consider the odd horizontal creases.
    By Lemma \ref{lem:fact2}, the type of fold ($90^\circ$ or $180^\circ$) along these creases must propagate over the whole width of the polyomino.

    If one of these folds is $90^\circ$, then consider the odd vertical creases.
    Clearly if these creases directly intersect the $90^\circ$ fold (or intersect at a trivially folding hole)
    then they must fold $180^\circ$ (by Lemma \ref{lem:rectangularSection}). If the intersection point is at a non-trivially folding hole, in the case where the 
    hole is horizontal, and thus in-line with the $90^\circ$ fold, then by Lemma \ref{lem:fact3}, the vertical crease cannot be folded $90^\circ$ and so must be folded $180^\circ$. If the intersection point is a hole perpendicular to the $90^\circ$ fold, then we have a $90^\circ$ folded crease passing 
    through the middle of a slit, and by Lemma \ref{lem:fact3} again, the vertical crease must be folded $180^\circ$. 
    Thus we may without loss of generality assume that every odd horizontal crease folds $180^\circ$. 

    The bottom row of squares (viewed in isolation) is a $1\times n$ strip, and so it can map to at most $4$ faces of the cube, in a ring shape. 
    Clearly, the presence of other squares will not change this fact.  
    Notice that this ring leaves at least two opposite faces of the cube uncovered,
    and that every vertex on odd  horizontal crease of this $1\times n$ strip (which may be the boundary of the polyomino) maps to the same side of this ring of $4$ faces.
    We call the yet uncovered face on this side of the cube the top of the cube, and the other face the bottom.

    Our goal from this point is to show that every square on the polyomino has at least one vertex which maps to the top of the cube. Clearly, it is enough to show that all vertices on the intersection of odd horizontal creases and odd vertical creases map to the top of the cube. We call these vertices \emph{odd-aligned}.

    We proceed by induction on the odd horizontal creases. The base case is trivial. We already know that the first odd horizontal crease maps to the top side of the ring, and thus all odd-aligned vertices in this crease do so, too.

    In the inductive step, assume that in the first $n$ odd horizontal creases, all odd-aligned vertices map to the top of the cube.
    Consider the even horizontal crease in between the $n$th and $(n+1)$th odd horizontal creases.
    For each odd aligned vertex on this crease, either the two squares at the top and bottom left or the two squares at the top and bottom right of this vertex must be attached to each other via an edge. Now, because there is a $180^\circ$ fold through this edge, the two squares we identified map to the same faces as each other. In particular, this means that the vertices above and below our vertex on the central crease map to the same point which must be at the top of the cube.

    So, each odd-aligned vertex maps to the same side of the ring. Now every square in the polyomino touches an odd-aligned vertex, as observed previously, and so at least one of the square's vertices maps to the top side of the ring. If a square was to map to the bottom face of the cube, then all four of its vertices would have to map to the bottom side of the ring. Thus no square maps to 
    the bottom face of the cube. 
\end{proof}

\begin{cor}
    \label{cor:bigslitsonly}
    Let $P$ be a rectangular polyomino containing only slit holes, such that its width and height are 
    both $\ge 6$. Then $P$ folds into a cube if and only if some two vertical or two horizontal holes are odd-separated.
\end{cor}
\begin{proof} 
    The reverse implication is shown in the above theorem, so it remains to 
    show that if $P$ is as stated above, then it folds.
    Throughout this proof we will assume without loss of generality that all slits are vertical.

    Note that if two slits are odd-separated with more than one row between them, we can apply $180^\circ$ leaving the centres of the holes separated by $3$ rows.
    
    The same can be done in the horizontal direction; for vertical slits with an odd number of columns inbetween  we reduce the horizontal separation to width $1$ or $3$, and for vertical slits with an even number of columns  we reduce the horizontal separation to width $2$. Hence $P$ can be reduced to one of the polyominoes shown in Figure \ref{fig:MinimalSlitPolys}, all of which are foldable.
    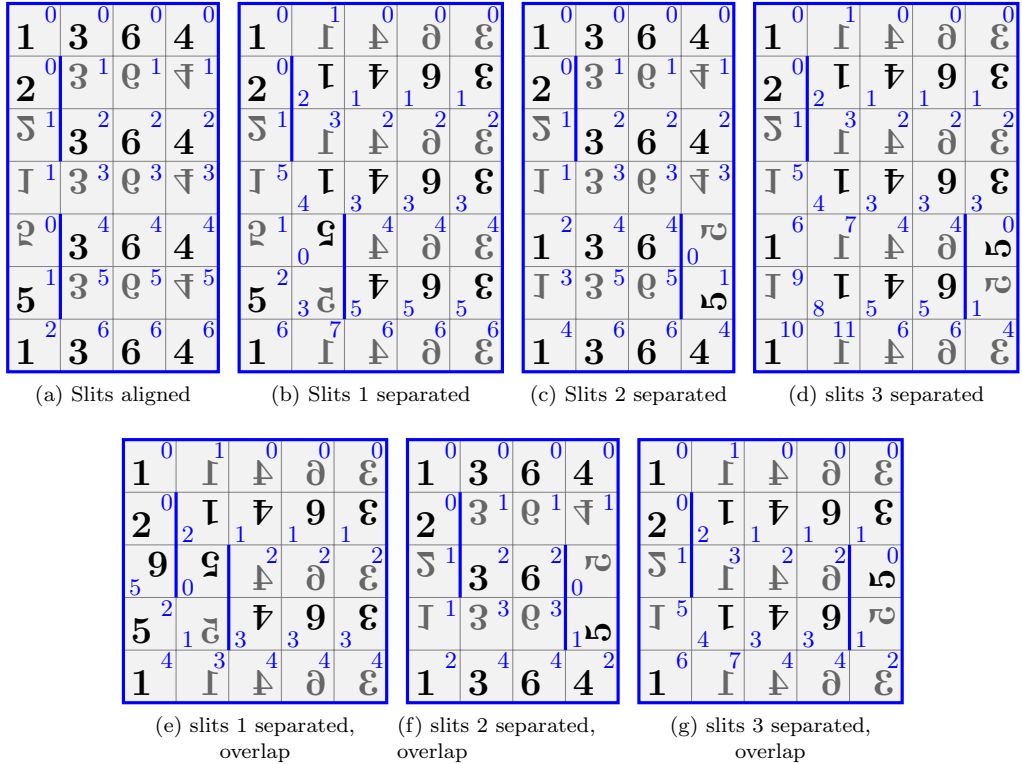
\begin{figure}
        \centering
        \input{slits_aligned}
        \caption{All minimal polyominoes and their foldings}
        \label{fig:MinimalSlitPolys}
    \end{figure}
    Note that without the $6\times 6$ requirement, it might be possible to position 
    holes in the same or adjacent creases (subfigures \ref{fig:MinimalSlitPolys} a, b, e)
    such that neither side of the holes has $3$ or more empty columns, which would therefore leave these particular hole arrangements non-foldable.
\end{proof}

The above corollary provides a complete classification of almost all polyominoes containing only 2-tall slits, thus leaving only the $k \times n$ cases for $k < 6$ open. Note that we can without loss of generality focus on the case where there is a pair of vertical odd-separated slits. Indeed, if there were odd-separated horizontal slits, rotating the polyomino by 90 degrees, we either end up with a $k \times n$-polyomino for $k < 6$ with odd separated vertical slits, or with polyomino which can be reduced to one of the polyominoes shown in Figure \ref{fig:MinimalSlitPolys}. 

Moreover, if two odd-separated slits were two or more columns apart, we could reduce the polyomino with only these holes to one of the cases in Figure \ref{fig:MinimalSlitPolys}, so this pair of holes cooperates. Thus, for the remainder of this section we assume all holes in odd-separated pairs are oriented vertically, and odd-separated pairs appear in the same crease or adjacent creases. 

We first consider the case where in addition to these odd separated vertical slits, there is also at least one horizontal slit. Note that in this case the width must be at least $4$ in order to accommodate the horizontal slit.

    \begin{lem}
        If a $4\times n$ or $5\times n$ polyomino $P$ containing only slits has two odd-separated vertical slits,
        and a horizontal slit, then it is foldable.
        \label{lem:3slitHorizontal}
    \end{lem}
    \begin{proof}
        Each of these polyominoes can be reduced to one of the minimal ones shown in Figure \ref{fig:slitsHorzontal}. For $4 \times n$-polyominoes, this can be achieved with a series of horizontal $180^\circ$ folds; for $5 \times n$ polyominoes we need an additional vertical fold which is always possible due to the fact that odd-separated vertical slits must appear in the central two creases as otherwise we could ignore the horizontal slit and reduce to one of the cases in Figure \ref{fig:MinimalSlitPolys}.
    \end{proof}

\input{horizontal_slits.tex}

From now on, we can thus assume that all slits are vertical. We start by showing that no $3 \times n$ polyomino is foldable.

\begin{lem}
    \label{lem:impossible-three-by-n}
Let $P$ be a rectangular $3 \times n$ polyomino whose only holes are vertical slits. Then $P$ does not fold.
\end{lem}

\begin{proof}
    By Lemma \ref{lem:fact1} above, all edges on each of the two vertical crease are folded the same (either $90 ^\circ$ or $180^\circ$) in a hypothetical folding onto the cube.
    Assume that such a folding exists.

    First assume that one of the vertical creases, without loss of generality the one between the central and right columns, is folded $180^\circ$. Then every vertex on the right boundary which is not at the same height as the midpoint of a slit on this crease maps to the same vertex of the cube as the vertex two units left of it.

    The central column of squares (viewed in isolation) is a $1\times n$ strip, and so it can map to at most $4$ faces of the cube, in a ring shape. Clearly, the presence of other squares will not change this fact.  Notice that this ring leaves at least two opposite faces $f$ and $f'$ of the cube uncovered and that every vertex on the left of this column maps to the same side, without loss of generality $f$, of this ring of $4$ faces. The above observation ensures that the image of every square in the right column has one vertex incident to $f$, moreover, every face in the left column shares at least one vertex with a face in the central column, and this vertex is mapped to a vertex of $f$. Hence every face of $P$ is mapped to a face which has a vertex with $f$ in common, and thus $f'$ is uncovered.

    Thus we may assume that all vertical edges are folded $90 ^\circ$. 
    But then the left and right column must map to the same ring of $4$ faces, the edges of both vertical creases must map to the same side, $f$, of this ring, 
    and a similar argument as before shows that the side, $f'$, opposite to $f$ remains uncovered.
\end{proof}

Since any $1 \times n$ and $2 \times n$ polyomino is contained in a $3 \times n$ polyomino, the above lemma shows that if a rectangular polyominoes with only slit holes is foldable, then its height and width are both at least $4$. 

\begin{lem}
    \label{lem:4byn-vertical-crease-pattern}
    Let $P$ be a $4 \times n$ polyomino with only vertical slit holes. Then in every folding of $P$ onto a cube either all vertical creases are folded $90^\circ$, or the central vertical crease is folded $90^\circ$ and the other two vertical creases are folded $180^\circ$.
\end{lem}

\begin{proof}
 We follow a similar approach as in the previous proof. If the central crease is folded $180^\circ$, then the central two columns map to the same ring of $4$ faces of the cube and every square of the other two columns contains at least one vertex whose image is incident to one side $f$ of this ring, so the opposite side remains uncovered. Similarly, if exactly one of the vertical creases, say the one between the first two columns, is folded $180^\circ$, then the first two columns and the fourth column all only cover faces in the same ring of $4$ faces. Each square in the third column contains at least one vertex whose image is incident to one side $f$ of this ring.
\end{proof}

\begin{lem}
    \label{lem:4byn-not-all-90}
    Let $P$ be a rectangular $4 \times n$ polyomino whose only holes are vertical slits. Assume that no two of these holes cooperate. Then $P$ has no folding onto the cube in which all vertical creases are folded $90^\circ$.
\end{lem}

\begin{proof}
Note that all slits in then non-central vertical creases must be even-separated, as otherwise two of them would cooperate. Since all vertical creases are folded $90^\circ$, the bottom row of the polyomino covers a ring of $4$ faces on the cube. 
Considering the left and right two columns of the polyomino as two distinct subpolyominoes, we can apply the proof of Theorem \ref{thm:slitsnonfolding} separately to the two $2 \times n$ polyominoes to see that every square of $P$ contains at least one vertex whose image is incident to the same side $f$ of this ring.
\end{proof}

\begin{lem}
    \label{lem:impossible-four-by-n}
    Let $P$ be a rectangular $4 \times n$ polyomino whose only holes are vertical slits. Assume that no two of these holes cooperate. Further assume that one of the off-central vertical creases contains no slit. Then $P$ has no folding onto the cube.
\end{lem}

\begin{proof}
    By Lemma \ref{lem:impossible-three-by-n}, there is no folding onto the cube where the crease with no hole is folded $180^\circ$. By Lemmas \ref{lem:4byn-vertical-crease-pattern} and \ref{lem:4byn-not-all-90}, there is no such folding where the crease is folded $90^\circ$.
\end{proof}

\begin{thm}
 If $P$ is a foldable $4 \times n$ polyomino with only vertical slit holes, then $P$ contains a cooperating set of at most $4$ holes.
\end{thm}

\begin{proof}
    We may without loss of generality assume that $P$ does not contain a cooperating pair of holes. Thus by Lemmas \ref{lem:4byn-vertical-crease-pattern} and \ref{lem:4byn-not-all-90}, in any folding onto the cube, the central vertical crease is folded $90^\circ$, and the other vertical creases are folded $180 ^\circ$. The $90^\circ$ fold in the central vertical crease implies that every second horizontal crease in the central two columns is folded $180^\circ$, thus each of these columns can only cover at most two faces of the cube.

    The $180^\circ$ fold between the left two columns means that these two columns map to a ring of $4$ faces, and thus the right two columns must cover the remaining two opposite faces. Since one of them only covers two different faces, this is only possible if at least one of the holes between these columns ring-folds. Analogously, we can show that some slit between the left two columns must ring-fold.

    If a slit between the left two columns ring-folds, then the crease running horizontally through the centre of this slit folds $90 ^\circ$. Since the central vertical crease also folds $90^{\circ}$, this implies that the top or bottom of a slit in the central vertical crease must be at the same height as the centre of the ring-folding slit.

    Hence there must be at least one slit in the crease between the left two columns which is $1$-separated from a slit in the central crease, and at least one slit in the crease between the right two columns which is $1$-separated from a slit in the central crease. The polyomino with only these holes can be reduced to one of the cases in Figure \ref{fig:4slitsminimal}, thus finishing the proof.
\end{proof}

        \input{fourslits}

\begin{lem}
    \label{lem:impossible-five-by-n}
Let $P$ be a rectangular $5 \times n$ polyomino whose only holes are vertical slits in the two central creases. Assume that all pairs of slits in the same crease are even separated. Then $P$ does not fold.
\end{lem}

\begin{proof}
By Lemma \ref{lem:impossible-four-by-n}, we may assume that the creases between the first two and between the last two columns are folded $90 \deg$. But then all horizontal creases must be folded $180 \deg$, so this folding cannot be onto.
\end{proof}

\begin{thm}
    \label{thm:5byn}
    A $5 \times n$ polyomino only containing vertical slits is foldable if and only if it contains a pair of cooperating slits.
\end{thm}
\begin{proof}
    The backwards implication is trivial. For the forward implication, assume that $P$ is a foldable $5 \times n$ polyomino only containing vertical slits.

    Label the $4$ vertical creases $a$-$d$, as seen in Figure \ref{fig:5bynCaseExample}. Recall that if a pair of odd-separated slits appears anywhere except in creases $b$ or $c$, then they guarantee foldability. Likewise, if two odd-separated slits appear in the same crease, we can reduce the polyomino to one of the foldable cases in Figure \ref{fig:MinimalSlitPolys}. In both of these cases we obtain a pair of cooperating holes.
    
    Thus if two holes are odd-separated, one must be in crease $b$, and the other in $c$. Note that if there are any slits in columns $a$ or $d$, they must be odd-separated with one of these holes, and so the polyomino contains a cooperating pair of holes. 
    
    Hence we may assume that $P$ only contains slits in creases $b$ and $c$, and all slits in crease $b$ are even separated, and all slits in crease $c$ are even separated. But such a polyomino is not foldable by Lemma \ref{lem:impossible-five-by-n}.
\end{proof}

    \begin{figure}
        \centering
        \begin{tikzpicture}
           \fill[gray!10] (0, 0) rectangle (5, 7.5);
           \draw[gray, very thin] (0, 0) grid (5, 7.5);
           \draw[blue, very thick] (0, 0) -- (5, 0);
           \draw[blue, very thick] (0, 0) -- (0, 7.5);
           \draw[blue, very thick] (5, 0) -- (5, 7.5);
           
           \draw[blue, very thick] (2, 1) -- (2, 3);
           \draw[blue, very thick] (3, 4) -- (3, 6);
           \draw[blue, very thick] (2, 5) -- (2, 7);

           \path (1, -0.2) node {a}
                (2, -0.2) node {b}
                (3, -0.2) node {c}
                (4, -0.2) node {d};
        \end{tikzpicture}
        \caption{Example $5\times n$ Case}
        \label{fig:5bynCaseExample}
    \end{figure}

\section{No Bound for Unit Holes and 2-Tall Slits}
\label{sec:unbounded}
    Given some $k\in\mathbb{N}$, construct a polyomino $P$ with $2k+1$ holes as follows:\
    Start with a $4$ by $4k+2$-wide rectangle, and add unit squares 
    to the top left and right corners. From the bottom right corner of the left hole,
    move down one square and place a $2$ wide slit horizontally from left to right.
    move one cell up from the right side of this hole, and place another hole, also
    moving right. Continue this pattern moving up and down until the other hole is reached and 
    all holes have been placed. See Figure \ref{fig:5holes} as an example. 
    \begin{figure}[hbt]
        \centering
        \input{example_5holes}
        \caption{An example polyomino with $5$ holes}
        \label{fig:5holes}
    \end{figure}

    \begin{thm}
        The polyomino $P$ described above is foldable, and all holes of $P$ must fold non-trivially in any folding onto the cube.

        In particular, for any $k\in\mathbb{N}$ there exists a polyomino with a minimally cooperating set of $2k+1$ holes.
    \end{thm}
    \begin{proof}
        Note that the facemapping in Figure \ref{fig:5holes} can be extended to larger examples in an obvious (periodic) way. It is not hard to show (inductively) that this gives a valid folding by applying all indicated folds at vertical creases in a $4 \times 4$ region containing two slits and no other holes and noticing that this reduces the polyomino to one with two fewer slits.

        It remains to show that in any consistent facemapping onto the cube, every hole must fold non-trivially.

        To this end, first note that if 
        both unit square holes fold trivially, we are left with a polyomino that only has even-separated slit holes.
        By theorem \ref{thm:slitsnonfolding}, this will not fold onto the cube. Thus at least one square hole (say, the left one) folds non-trivially.

        Next we look at possible foldings of the slit holes. Examining the possible foldings of the square hole, we note that there is no non-trivial folding in which creases $a$ and $b$ in Figure \ref{fig:unboundedWithRing} are both folded $180^\circ$. 

        If one of the slits in crease $a$ is folded into a ring, then 
        by Lemma \ref{lem:fact3}, crease $a$ must fold $180^\circ$.
        Additionally, because the region of $P$ above any slit in crease $a$ is a $2\times 2$ rectangle,
        and because a ring-fold will force a vertical $90^\circ$ fold above it,
        crease $b$ must be folded $180^\circ$ too (via Lemma \ref{lem:rectangularSection}). Since $180^\circ$,  both $a$ and $b$ would be folded $180^\circ$ which is impossible.

        If one of the lower slits is folded into a ring, then the two columns directly above it contain no holes
        and so again, both rows must be folded
        $180^\circ$. Just like in the previous case,
        we therefore know that both creases $a$ and $b$ must be folded $180^\circ$.
        Thus in any folding onto the cube, all slits must either fold to make flaps, or fold trivially.

        \begin{figure}
            \centering
            \begin{tikzpicture}
                \fill[gray!10] (0, 0.1) rectangle (2.9, 3);
                \draw[gray, very thin] (0, 0.1) grid (2.9, 3);
                \draw[blue, very thick] (0, 0.1) -- (0, 3);
                \draw[blue, very thick] (0, 3) -- (2.9, 3);
                \draw[blue, very thick] (1, 1) rectangle (2, 2);
                \path (3, 1) node {a}
                    (3, 2) node {b};
            \end{tikzpicture}
            \caption{Unit holes when a slit is ring-folded}
            \label{fig:unboundedWithRing}
        \end{figure}

        Accounting for different rotations, there are four non-trivial ways of folding the left unit square hole. Our next goal is to show that only one of them has a chance of being extended to a folding of $P$ onto the cube.

        Let us call a vertical crease \emph{odd} if it goes through the centre of a slit (or if there is an even number of columns between it and the centre of a slit), and \emph{even} otherwise. All odd vertical creases which lie strictly between the two unit square holes are folded $180^\circ$, and thus all vertices in even vertical creases are mapped to the same vertex of the cube as the vertex which lies in the same horizontal crease and the even crease touching the left unit square hole.

        Note that any non-trivial folding of the left unit square hole covers all faces of the cube but one; let us denote this final face by $f$. 
        The image of any vertex on the inside of the left unit square hole does not border $f$ (otherwise the U-hole would have a folding onto the cube), and thus no face in the top three rows which borders an even crease can map to $f$.
        Faces in the rightmost column do not border an even crease, but as both $180^\circ$ and $90^\circ$ folds propagate horizontally through slits, the right unit square hole can be folded in precisely the same ways as it could if all slits were folded trivially. Therefore, no face in the top three rows can cover $f$.

        If the bottom vertical crease was folded $180^\circ$, then the same argument as above would show that no square in the lowest row can cover $f$. Thus, we conclude that the bottom vertical crease must be folded $90^\circ$. Applying Lemma \ref{lem:rectangularSection} to the bottom left $2 \times 2$ square, we see that the vertical creases in this square must be folded $90^\circ$, leaving exactly one possible non-trivial folding for the left unit square hole.

        \begin{figure}[htb]
            \centering
            \input{thm3_good_case}
            \caption{final partial facemapping with $90^\circ$ fold}
            \label{fig:thm3FinalGood}
        \end{figure}
        The partial folding we obtain is shown in Figure \ref{fig:thm3FinalGood}. It remains to show that if we extend this folding, then each flap must fold non-trivially.
        We label the squares relevant to the first three slits a through $c$ as in Figure \ref{fig:thm3FinalGood}, and will use induction to show that each hole must be folded non-trivially. 
        
        Consider the $90^\circ$ folded crease directly above the leftmost
        green crease. We aim to show that if this crease is folded $90^\circ$, then the green crease $4$ columns to the left of it must also be folded $90^\circ$, and all slits in between these two creases must fold non-trivially.

        If the slit above the faces labelled $a$ was folded trivially, then
        we would have two $90^\circ$ folded creases intersecting with each 
        other in a trivial region, which is impossible.
        So the slit must fold non-trivially.

        Next, assume that the slit above the squares labeled $b$ was folded trivially.
        This would mean that both green creases directly below it would be folded $180^\circ$
        as the creases above flap $b$ are folded $180^\circ$. This would imply that the top two faces labelled $b$ would map to face $4$, and the bottom two to face $5$.

        Regardless of how the slit above the faces labelled $a$ is folded,
        the vertices marked "$*$" and "$\star$" in Figure \ref{fig:thm3FinalGood} must map to the same vertex on the cube.
        However, the square with the $*$ is folded to face $5$ from the edge of face $6$, and the square with the $\star$ is folded to face $5$ 
        from the edge of face $4$. This means that $*$ maps to the corner of $5$, $4$ and $1$, while $\star$ maps to $5$, $6$ and $3$ 
        %(seen in Figure \ref{fig:badVertexMapping}). 
        This is a contradiction, and so the slit above the faces labelled $b$ is folded non-trivially too.

        Finally, observe that if both of the $b$ faces in the top row mapped to face $4$, then the square marked, we would arrive at the same contradiction as before regarding the vertices marked "$*$" and "$\star$". So the squares cannot map to face $4$ and so the green crease is folded $90^\circ$.
        This completes our induction step, showing that every such crease is folded $90^\circ$, and also that all slits must be folded non-trivially.

        Notice that the final two columns in the polyomino
        (containing the second unit square) occur after a slit in the lowest crease. If these two columns were folded trivially,
        then we can use the same logic as above to derive a contradiction, so the last unit hole must be folded non-trivially too.
    \end{proof}
    
\bibliography{sources}
\bibliographystyle{plain}

\end{document}

%% file: Lshaped_square.tex
\begin{tikzpicture}
    \fill[gray!10] (0, 0) rectangle (4, 4);
    \draw[gray, very thin] (0, 0) grid (4, 4);
    \draw[blue, very thick] (0, 0) rectangle (4, 4);
    \draw[blue, very thick] (2, 2) rectangle (3, 3);
    \draw[blue, very thick] (1, 1) -- (2, 1);
    \draw[blue, very thick] (1, 1) -- (1, 2);
\end{tikzpicture}

%% file: Lshaped_Ushaped.tex
\begin{tikzpicture}
    \fill[gray!10] (0, 0) rectangle (4, 4);
    \draw[gray, very thin] (0, 0) grid (4, 4);
    \draw[blue, very thick] (0, 0) rectangle (4, 4);
    \draw[blue, very thick] (1, 1) -- (2, 1);
    \draw[blue, very thick] (1, 1) -- (1, 2);
    \draw[blue, very thick] (2, 2) -- (2, 3);
    \draw[blue, very thick] (2, 2) -- (3, 2);
    \draw[blue, very thick] (3, 2) -- (3, 3);
\end{tikzpicture}

%% file: slits_aligned.tex
\subfloat[Slits aligned]{
            \begin{tikzpicture}[xscale=0.7,yscale=0.7]
                \initcube\printcube
                \rollup\printcube
                \flipup\printcube
                \rollup\printcube
                \rollup\printcube
                \flipup\printcube
                \rollup\printcube
                \rollright\printcube
                \flipdown\printcube
                \flipdown\printcube
                \flipdown\printcube
                \flipdown\printcube
                \flipdown\printcube
                \flipdown\printcube
                \rollright\printcube
                \flipup\printcube
                \flipup\printcube
                \flipup\printcube
                \flipup\printcube
                \flipup\printcube
                \flipup\printcube
                \rollright\printcube
                \flipdown\printcube
                \flipdown\printcube
                \flipdown\printcube
                \flipdown\printcube
                \flipdown\printcube
                \flipdown\printcube
                
                \begin{scope}[shift={(-.5,-.5)}]
                    \draw[gray, very thin] (0, 0) grid (4, 7); 
                    \draw[blue, very thick] (0, 0) rectangle (4, 7);
                    \draw[blue, very thick] (1, 1) -- (1, 3);
                    \draw[blue, very thick] (1, 4) -- (1, 6);   

                    \node[anchor=east,inner sep=1pt, blue] at (1,0.8) {2};
                    \node[anchor=east,inner sep=1pt, blue] at (2,0.8) {6};
                    \node[anchor=east,inner sep=1pt, blue] at (3,0.8) {6};
                    \node[anchor=east,inner sep=1pt, blue] at (4,0.8) {6};

                    \node[anchor=east,inner sep=1pt, blue] at (1,1.8) {1};
                    \node[anchor=east,inner sep=1pt, blue] at (2,1.8) {5};
                    \node[anchor=east,inner sep=1pt, blue] at (3,1.8) {5};
                    \node[anchor=east,inner sep=1pt, blue] at (4,1.8) {5};

                    \node[anchor=east,inner sep=1pt, blue] at (1,2.8) {0};
                    \node[anchor=east,inner sep=1pt, blue] at (2,2.8) {4};
                    \node[anchor=east,inner sep=1pt, blue] at (3,2.8) {4};
                    \node[anchor=east,inner sep=1pt, blue] at (4,2.8) {4};

                    \node[anchor=east,inner sep=1pt, blue] at (1,3.8) {1};
                    \node[anchor=east,inner sep=1pt, blue] at (2,3.8) {3};
                    \node[anchor=east,inner sep=1pt, blue] at (3,3.8) {3};
                    \node[anchor=east,inner sep=1pt, blue] at (4,3.8) {3};

                    \node[anchor=east,inner sep=1pt, blue] at (1,4.8) {1};
                    \node[anchor=east,inner sep=1pt, blue] at (2,4.8) {2};
                    \node[anchor=east,inner sep=1pt, blue] at (3,4.8) {2};
                    \node[anchor=east,inner sep=1pt, blue] at (4,4.8) {2};

                    \node[anchor=east,inner sep=1pt, blue] at (1,5.8) {0};
                    \node[anchor=east,inner sep=1pt, blue] at (2,5.8) {1};
                    \node[anchor=east,inner sep=1pt, blue] at (3,5.8) {1};
                    \node[anchor=east,inner sep=1pt, blue] at (4,5.8) {1};

                    \node[anchor=east,inner sep=1pt, blue] at (1,6.8) {0};
                    \node[anchor=east,inner sep=1pt, blue] at (2,6.8) {0};
                    \node[anchor=east,inner sep=1pt, blue] at (3,6.8) {0};
                    \node[anchor=east,inner sep=1pt, blue] at (4,6.8) {0};

                \end{scope}
            \end{tikzpicture}
        }
        \subfloat[Slits 1 separated]{
            \begin{tikzpicture}[xscale=0.7,yscale=0.7]
                \initcube\printcube
                \rollup\printcube
                \flipup\printcube
                \rollup\printcube
                \rollup\printcube
                \flipup\printcube
                \rollup\printcube
                \flipright\printcube
                \flipdown\printcube
                \flipdown\printcube
                \flipdown\printcube
                \rolldown\printcube
                \flipdown\printcube
                \rolldown\printcube
                \rollright\printcube
                \flipup\printcube
                \flipup\printcube
                \flipup\printcube
                \flipup\printcube
                \flipup\printcube
                \flipup\printcube
                \rollright\printcube
                \flipdown\printcube
                \flipdown\printcube
                \flipdown\printcube
                \flipdown\printcube
                \flipdown\printcube
                \flipdown\printcube
                \rollright\printcube
                \flipup\printcube
                \flipup\printcube
                \flipup\printcube
                \flipup\printcube
                \flipup\printcube
                \flipup\printcube
                
                \begin{scope}[shift={(-.5,-.5)}]
                    \draw[gray, very thin] (0, 0) grid (5, 7); 
                    \draw[blue, very thick] (0, 0) rectangle (5, 7);
                    \draw[blue, very thick] (2, 1) -- (2, 3);
                    \draw[blue, very thick] (1, 4) -- (1, 6);   

                    \node[anchor=east,inner sep=1pt, blue] at (1,0.8) {6};
                    \node[anchor=east,inner sep=1pt, blue] at (2,0.8) {7};
                    \node[anchor=east,inner sep=1pt, blue] at (3,0.8) {6};
                    \node[anchor=east,inner sep=1pt, blue] at (4,0.8) {6};
                    \node[anchor=east,inner sep=1pt, blue] at (5,0.8) {6};

                    \node[anchor=east,inner sep=1pt, blue] at (1,1.8) {2};
                    \node[anchor=east,inner sep=1pt, blue] at (1.4,1.2) {3};
                    \node[anchor=east,inner sep=1pt, blue] at (2.4,1.2) {5};
                    \node[anchor=east,inner sep=1pt, blue] at (3.4,1.2) {5};
                    \node[anchor=east,inner sep=1pt, blue] at (4.4,1.2) {5};

                    \node[anchor=east,inner sep=1pt, blue] at (1,2.8) {1};
                    \node[anchor=east,inner sep=1pt, blue] at (1.4,2.2) {0};
                    \node[anchor=east,inner sep=1pt, blue] at (3,2.8) {4};
                    \node[anchor=east,inner sep=1pt, blue] at (4,2.8) {4};
                    \node[anchor=east,inner sep=1pt, blue] at (5,2.8) {4};

                    \node[anchor=east,inner sep=1pt, blue] at (1,3.8) {5};
                    \node[anchor=east,inner sep=1pt, blue] at (1.4,3.2) {4};
                    \node[anchor=east,inner sep=1pt, blue] at (2.4,3.2) {3};
                    \node[anchor=east,inner sep=1pt, blue] at (3.4,3.2) {3};
                    \node[anchor=east,inner sep=1pt, blue] at (4.4,3.2) {3};

                    \node[anchor=east,inner sep=1pt, blue] at (1,4.8) {1};
                    \node[anchor=east,inner sep=1pt, blue] at (2,4.8) {3};
                    \node[anchor=east,inner sep=1pt, blue] at (3,4.8) {2};
                    \node[anchor=east,inner sep=1pt, blue] at (4,4.8) {2};
                    \node[anchor=east,inner sep=1pt, blue] at (5,4.8) {2};

                    \node[anchor=east,inner sep=1pt, blue] at (1,5.8) {0};
                    \node[anchor=east,inner sep=1pt, blue] at (1.4,5.2) {2};
                    \node[anchor=east,inner sep=1pt, blue] at (2.4,5.2) {1};
                    \node[anchor=east,inner sep=1pt, blue] at (3.4,5.2) {1};
                    \node[anchor=east,inner sep=1pt, blue] at (4.4,5.2) {1};

                    \node[anchor=east,inner sep=1pt, blue] at (1,6.8) {0};
                    \node[anchor=east,inner sep=1pt, blue] at (2,6.8) {1};
                    \node[anchor=east,inner sep=1pt, blue] at (3,6.8) {0};
                    \node[anchor=east,inner sep=1pt, blue] at (4,6.8) {0};
                    \node[anchor=east,inner sep=1pt, blue] at (5,6.8) {0};

                \end{scope}
            \end{tikzpicture}
        }
        \subfloat[Slits 2 separated]{
            \begin{tikzpicture}[xscale=0.7,yscale=0.7]
                \initcube\printcube
                \flipup\printcube
                \flipup\printcube
                \flipup\printcube
                \rollup\printcube
                \flipup\printcube
                \rollup\printcube
                \rollright\printcube
                \flipdown\printcube
                \flipdown\printcube
                \flipdown\printcube
                \flipdown\printcube
                \flipdown\printcube
                \flipdown\printcube
                \rollright\printcube
                \flipup\printcube
                \flipup\printcube
                \flipup\printcube
                \flipup\printcube
                \flipup\printcube
                \flipup\printcube
                \rollright\printcube
                \flipdown\printcube
                \flipdown\printcube
                \flipdown\printcube
                \rolldown\printcube
                \flipdown\printcube
                \rolldown\printcube

                \begin{scope}[shift={(-.5,-.5)}]
                    \draw[gray, very thin] (0, 0) grid (4, 7); 
                    \draw[blue, very thick] (0, 0) rectangle (4, 7);
                    \draw[blue, very thick] (3, 1) -- (3, 3);
                    \draw[blue, very thick] (1, 4) -- (1, 6);   

                    \node[anchor=east,inner sep=1pt, blue] at (1,0.8) {4};
                    \node[anchor=east,inner sep=1pt, blue] at (2,0.8) {6};
                    \node[anchor=east,inner sep=1pt, blue] at (3,0.8) {6};
                    \node[anchor=east,inner sep=1pt, blue] at (4,0.8) {4};

                    \node[anchor=east,inner sep=1pt, blue] at (1,1.8) {3};
                    \node[anchor=east,inner sep=1pt, blue] at (2,1.8) {5};
                    \node[anchor=east,inner sep=1pt, blue] at (3,1.8) {5};
                    \node[anchor=east,inner sep=1pt, blue] at (4,1.8) {1};

                    \node[anchor=east,inner sep=1pt, blue] at (1,2.8) {2};
                    \node[anchor=east,inner sep=1pt, blue] at (2,2.8) {4};
                    \node[anchor=east,inner sep=1pt, blue] at (3,2.8) {4};
                    \node[anchor=east,inner sep=1pt, blue] at (3.4,2.2) {0};

                    \node[anchor=east,inner sep=1pt, blue] at (1,3.8) {1};
                    \node[anchor=east,inner sep=1pt, blue] at (2,3.8) {3};
                    \node[anchor=east,inner sep=1pt, blue] at (3,3.8) {3};
                    \node[anchor=east,inner sep=1pt, blue] at (4,3.8) {3};

                    \node[anchor=east,inner sep=1pt, blue] at (1,4.8) {1};
                    \node[anchor=east,inner sep=1pt, blue] at (2,4.8) {2};
                    \node[anchor=east,inner sep=1pt, blue] at (3,4.8) {2};
                    \node[anchor=east,inner sep=1pt, blue] at (4,4.8) {2};

                    \node[anchor=east,inner sep=1pt, blue] at (1,5.8) {0};
                    \node[anchor=east,inner sep=1pt, blue] at (2,5.8) {1};
                    \node[anchor=east,inner sep=1pt, blue] at (3,5.8) {1};
                    \node[anchor=east,inner sep=1pt, blue] at (4,5.8) {1};

                    \node[anchor=east,inner sep=1pt, blue] at (1,6.8) {0};
                    \node[anchor=east,inner sep=1pt, blue] at (2,6.8) {0};
                    \node[anchor=east,inner sep=1pt, blue] at (3,6.8) {0};
                    \node[anchor=east,inner sep=1pt, blue] at (4,6.8) {0};

                \end{scope}
            \end{tikzpicture}
        }
        \subfloat[slits 3 separated]{
            \begin{tikzpicture}[xscale=0.7,yscale=0.7]
                \initcube\printcube
                \flipup\printcube
                \flipup\printcube
                \flipup\printcube
                \rollup\printcube
                \flipup\printcube
                \rollup\printcube
                \flipright\printcube
                \flipdown\printcube
                \flipdown\printcube
                \flipdown\printcube
                \flipdown\printcube
                \flipdown\printcube
                \flipdown\printcube
                \rollright\printcube
                \flipup\printcube
                \flipup\printcube
                \flipup\printcube
                \flipup\printcube
                \flipup\printcube
                \flipup\printcube
                \rollright\printcube
                \flipdown\printcube
                \flipdown\printcube
                \flipdown\printcube
                \flipdown\printcube
                \flipdown\printcube
                \flipdown\printcube
                \rollright\printcube
                \rollup\printcube
                \flipup\printcube
                \rollup\printcube
                \flipup\printcube
                \flipup\printcube
                \flipup\printcube

                \begin{scope}[shift={(-.5,-.5)}]
                    \draw[gray, very thin] (0, 0) grid (5, 7); 
                    \draw[blue, very thick] (0, 0) rectangle (5, 7);
                    \draw[blue, very thick] (4, 1) -- (4, 3);
                    \draw[blue, very thick] (1, 4) -- (1, 6);   

                    \node[anchor=east,inner sep=1pt, blue] at (1,0.8) {10};
                    \node[anchor=east,inner sep=1pt, blue] at (2,0.8) {11};
                    \node[anchor=east,inner sep=1pt, blue] at (3,0.8) {6};
                    \node[anchor=east,inner sep=1pt, blue] at (4,0.8) {6};
                    \node[anchor=east,inner sep=1pt, blue] at (5,0.8) {4};

                    \node[anchor=east,inner sep=1pt, blue] at (1,1.8) {9};
                    \node[anchor=east,inner sep=1pt, blue] at (1.4,1.2) {8};
                    \node[anchor=east,inner sep=1pt, blue] at (2.4,1.2) {5};
                    \node[anchor=east,inner sep=1pt, blue] at (3.4,1.2) {5};
                    \node[anchor=east,inner sep=1pt, blue] at (4.4,1.2) {1};

                    \node[anchor=east,inner sep=1pt, blue] at (1,2.8) {6};
                    \node[anchor=east,inner sep=1pt, blue] at (2,2.8) {7};
                    \node[anchor=east,inner sep=1pt, blue] at (3,2.8) {4};
                    \node[anchor=east,inner sep=1pt, blue] at (4,2.8) {4};
                    \node[anchor=east,inner sep=1pt, blue] at (5,2.8) {0};

                    \node[anchor=east,inner sep=1pt, blue] at (1,3.8) {5};
                    \node[anchor=east,inner sep=1pt, blue] at (1.4,3.2) {4};
                    \node[anchor=east,inner sep=1pt, blue] at (2.4,3.2) {3};
                    \node[anchor=east,inner sep=1pt, blue] at (3.4,3.2) {3};
                    \node[anchor=east,inner sep=1pt, blue] at (4.4,3.2) {3};

                    \node[anchor=east,inner sep=1pt, blue] at (1,4.8) {1};
                    \node[anchor=east,inner sep=1pt, blue] at (2,4.8) {3};
                    \node[anchor=east,inner sep=1pt, blue] at (3,4.8) {2};
                    \node[anchor=east,inner sep=1pt, blue] at (4,4.8) {2};
                    \node[anchor=east,inner sep=1pt, blue] at (5,4.8) {2};

                    \node[anchor=east,inner sep=1pt, blue] at (1,5.8) {0};
                    \node[anchor=east,inner sep=1pt, blue] at (1.4,5.2) {2};
                    \node[anchor=east,inner sep=1pt, blue] at (2.4,5.2) {1};
                    \node[anchor=east,inner sep=1pt, blue] at (3.4,5.2) {1};
                    \node[anchor=east,inner sep=1pt, blue] at (4.4,5.2) {1};

                    \node[anchor=east,inner sep=1pt, blue] at (1,6.8) {0};
                    \node[anchor=east,inner sep=1pt, blue] at (2,6.8) {1};
                    \node[anchor=east,inner sep=1pt, blue] at (3,6.8) {0};
                    \node[anchor=east,inner sep=1pt, blue] at (4,6.8) {0};
                    \node[anchor=east,inner sep=1pt, blue] at (5,6.8) {0};

                \end{scope}
            \end{tikzpicture}
        }
        \\
        \subfloat[slits 1 separated,\\ overlap]{
            \begin{tikzpicture}[xscale=0.7,yscale=0.7]
                \initcube\printcube
                \rollup\printcube
                \rollup\printcube
                \rollup\printcube
                \rollup\printcube
                \flipright\printcube
                \flipdown\printcube
                \rolldown\printcube
                \flipdown\printcube
                \rolldown\printcube
                \rollright\printcube
                \flipup\printcube
                \flipup\printcube
                \flipup\printcube
                \flipup\printcube
                \rollright\printcube
                \flipdown\printcube
                \flipdown\printcube
                \flipdown\printcube
                \flipdown\printcube
                \rollright\printcube
                \flipup\printcube
                \flipup\printcube
                \flipup\printcube
                \flipup\printcube

                \begin{scope}[shift={(-.5,-.5)}]
                    \draw[gray, very thin] (0, 0) grid (5, 5); 
                    \draw[blue, very thick] (0, 0) rectangle (5, 5);
                    \draw[blue, very thick] (2, 1) -- (2, 3);
                    \draw[blue, very thick] (1, 2) -- (1, 4);   

                    \node[anchor=east,inner sep=1pt, blue] at (1,0.8) {4};
                    \node[anchor=east,inner sep=1pt, blue] at (2,0.8) {3};
                    \node[anchor=east,inner sep=1pt, blue] at (3,0.8) {4};
                    \node[anchor=east,inner sep=1pt, blue] at (4,0.8) {4};
                    \node[anchor=east,inner sep=1pt, blue] at (5,0.8) {4};

                    \node[anchor=east,inner sep=1pt, blue] at (1,1.8) {2};
                    \node[anchor=east,inner sep=1pt, blue] at (1.4,1.2) {1};
                    \node[anchor=east,inner sep=1pt, blue] at (2.4,1.2) {3};
                    \node[anchor=east,inner sep=1pt, blue] at (3.4,1.2) {3};
                    \node[anchor=east,inner sep=1pt, blue] at (4.4,1.2) {3};

                    \node[anchor=east,inner sep=1pt, blue] at (0.4,2.2) {5};
                    \node[anchor=east,inner sep=1pt, blue] at (1.4,2.2) {0};
                    \node[anchor=east,inner sep=1pt, blue] at (3,2.8) {2};
                    \node[anchor=east,inner sep=1pt, blue] at (4,2.8) {2};
                    \node[anchor=east,inner sep=1pt, blue] at (5,2.8) {2};

                    \node[anchor=east,inner sep=1pt, blue] at (1,3.8) {0};
                    \node[anchor=east,inner sep=1pt, blue] at (1.4,3.2) {2};
                    \node[anchor=east,inner sep=1pt, blue] at (2.4,3.2) {1};
                    \node[anchor=east,inner sep=1pt, blue] at (3.4,3.2) {1};
                    \node[anchor=east,inner sep=1pt, blue] at (4.4,3.2) {1};

                    \node[anchor=east,inner sep=1pt, blue] at (1,4.8) {0};
                    \node[anchor=east,inner sep=1pt, blue] at (2,4.8) {1};
                    \node[anchor=east,inner sep=1pt, blue] at (3,4.8) {0};
                    \node[anchor=east,inner sep=1pt, blue] at (4,4.8) {0};
                    \node[anchor=east,inner sep=1pt, blue] at (5,4.8) {0};

                \end{scope}
            \end{tikzpicture}
        }
        \subfloat[slits 2 separated,\\ overlap]{
            \begin{tikzpicture}[xscale=0.7,yscale=0.7]
                \initcube\printcube
                \flipup\printcube
                \rollup\printcube
                \flipup\printcube
                \rollup\printcube
                \rollright\printcube
                \flipdown\printcube
                \flipdown\printcube
                \flipdown\printcube
                \flipdown\printcube
                \rollright\printcube
                \flipup\printcube
                \flipup\printcube
                \flipup\printcube
                \flipup\printcube
                \rollright\printcube
                \flipdown\printcube
                \rolldown\printcube
                \flipdown\printcube
                \rolldown\printcube

                \begin{scope}[shift={(-.5,-.5)}]
                    \draw[gray, very thin] (0, 0) grid (4, 5); 
                    \draw[blue, very thick] (0, 0) rectangle (4, 5);
                    \draw[blue, very thick] (3, 1) -- (3, 3);
                    \draw[blue, very thick] (1, 2) -- (1, 4);   

                    \node[anchor=east,inner sep=1pt, blue] at (1,0.8) {2};
                    \node[anchor=east,inner sep=1pt, blue] at (2,0.8) {4};
                    \node[anchor=east,inner sep=1pt, blue] at (3,0.8) {4};
                    \node[anchor=east,inner sep=1pt, blue] at (4,0.8) {2};

                    \node[anchor=east,inner sep=1pt, blue] at (1,1.8) {1};
                    \node[anchor=east,inner sep=1pt, blue] at (2,1.8) {3};
                    \node[anchor=east,inner sep=1pt, blue] at (3,1.8) {3};
                    \node[anchor=east,inner sep=1pt, blue] at (3.4,1.2) {1};

                    \node[anchor=east,inner sep=1pt, blue] at (1,2.8) {1};
                    \node[anchor=east,inner sep=1pt, blue] at (2,2.8) {2};
                    \node[anchor=east,inner sep=1pt, blue] at (3,2.8) {2};
                    \node[anchor=east,inner sep=1pt, blue] at (3.4,2.2) {0};

                    \node[anchor=east,inner sep=1pt, blue] at (1,3.8) {0};
                    \node[anchor=east,inner sep=1pt, blue] at (2,3.8) {1};
                    \node[anchor=east,inner sep=1pt, blue] at (3,3.8) {1};
                    \node[anchor=east,inner sep=1pt, blue] at (4,3.8) {1};

                    \node[anchor=east,inner sep=1pt, blue] at (1,4.8) {0};
                    \node[anchor=east,inner sep=1pt, blue] at (2,4.8) {0};
                    \node[anchor=east,inner sep=1pt, blue] at (3,4.8) {0};
                    \node[anchor=east,inner sep=1pt, blue] at (4,4.8) {0};

                \end{scope}
            \end{tikzpicture}
        }
        \subfloat[slits 3 separated,\\ overlap]{
            \begin{tikzpicture}[xscale=0.7,yscale=0.7]
                \initcube\printcube
                \flipup\printcube
                \rollup\printcube
                \flipup\printcube
                \rollup\printcube
                \flipright\printcube
                \flipdown\printcube
                \flipdown\printcube
                \flipdown\printcube
                \flipdown\printcube
                \rollright\printcube
                \flipup\printcube
                \flipup\printcube
                \flipup\printcube
                \flipup\printcube
                \rollright\printcube
                \flipdown\printcube
                \flipdown\printcube
                \flipdown\printcube
                \flipdown\printcube
                \rollright\printcube
                \rollup\printcube
                \flipup\printcube
                \rollup\printcube
                \flipup\printcube

                \begin{scope}[shift={(-.5,-.5)}]
                    \draw[gray, very thin] (0, 0) grid (5, 5); 
                    \draw[blue, very thick] (0, 0) rectangle (5, 5);
                    \draw[blue, very thick] (4, 1) -- (4, 3);
                    \draw[blue, very thick] (1, 2) -- (1, 4);   

                    \node[anchor=east,inner sep=1pt, blue] at (1,0.8) {6};
                    \node[anchor=east,inner sep=1pt, blue] at (2,0.8) {7};
                    \node[anchor=east,inner sep=1pt, blue] at (3,0.8) {4};
                    \node[anchor=east,inner sep=1pt, blue] at (4,0.8) {4};
                    \node[anchor=east,inner sep=1pt, blue] at (5,0.8) {2};

                    \node[anchor=east,inner sep=1pt, blue] at (1,1.8) {5};
                    \node[anchor=east,inner sep=1pt, blue] at (1.4,1.2) {4};
                    \node[anchor=east,inner sep=1pt, blue] at (2.4,1.2) {3};
                    \node[anchor=east,inner sep=1pt, blue] at (3.4,1.2) {3};
                    \node[anchor=east,inner sep=1pt, blue] at (4.4,1.2) {1};

                    \node[anchor=east,inner sep=1pt, blue] at (1,2.8) {1};
                    \node[anchor=east,inner sep=1pt, blue] at (2,2.8) {3};
                    \node[anchor=east,inner sep=1pt, blue] at (3,2.8) {2};
                    \node[anchor=east,inner sep=1pt, blue] at (4,2.8) {2};
                    \node[anchor=east,inner sep=1pt, blue] at (5,2.8) {0};

                    \node[anchor=east,inner sep=1pt, blue] at (1,3.8) {0};
                    \node[anchor=east,inner sep=1pt, blue] at (1.4,3.2) {2};
                    \node[anchor=east,inner sep=1pt, blue] at (2.4,3.2) {1};
                    \node[anchor=east,inner sep=1pt, blue] at (3.4,3.2) {1};
                    \node[anchor=east,inner sep=1pt, blue] at (4.4,3.2) {1};

                    \node[anchor=east,inner sep=1pt, blue] at (1,4.8) {0};
                    \node[anchor=east,inner sep=1pt, blue] at (2,4.8) {1};
                    \node[anchor=east,inner sep=1pt, blue] at (3,4.8) {0};
                    \node[anchor=east,inner sep=1pt, blue] at (4,4.8) {0};
                    \node[anchor=east,inner sep=1pt, blue] at (5,4.8) {0};

                \end{scope}
            \end{tikzpicture}
        }

%% file: horizontal_slits.tex
\begin{figure}
    \centering
    \subfloat{
        \begin{tikzpicture}[xscale=0.7,yscale=0.7]
            \initcube\printcube
            \flipup\printcube
            \flipup\printcube
            \flipup\printcube
            \flipup\printcube
            \rollup\printcube
            \flipup\printcube
            \rollup\printcube
            \flipright\printcube
            \rolldown\printcube
            \flipdown\printcube
            \rolldown\printcube
            \flipdown\printcube
            \flipdown\printcube
            \flipdown\printcube
            \rollright\printcube
            \rollup\printcube
            \flipup\printcube
            \rollup\printcube
            \flipup\printcube
            \flipup\printcube
            \flipup\printcube
            \flipright\printcube
            \flipdown\printcube
            \flipdown\printcube
            \flipdown\printcube
            \rolldown\printcube
            \flipdown\printcube
            \rolldown\printcube
            \flipdown\printcube
            \rollleft\printcube
            \rollleft\printcube
            \begin{scope}[shift={(-.5,-.5)}]
                \draw[gray, very thin] (0, 0) grid (4, 8); 
                \draw[blue, very thick] (0, 0) rectangle (4, 8);
                \draw[blue, very thick] (1, 1) -- (3, 1);
                \draw[blue, very thick] (2, 2) -- (2, 4);   
                \draw[blue, very thick] (2, 5) -- (2, 7);   
                
                \node[anchor=east,inner sep=1pt, blue] at (1,0.8) {10};
                \node[anchor=east,inner sep=1pt, blue] at (2,0.8) {0};
                \node[anchor=east,inner sep=1pt, blue] at (3,0.8) {0};
                \node[anchor=east,inner sep=1pt, blue] at (4,0.8) {10};

                \node[anchor=east,inner sep=1pt, blue] at (1,1.8) {9};
                \node[anchor=east,inner sep=1pt, blue] at (1.4,1.2) {0};
                \node[anchor=east,inner sep=1pt, blue] at (2.4,1.2) {0};
                \node[anchor=east,inner sep=1pt, blue] at (4,1.8) {9};

                \node[anchor=east,inner sep=1pt, blue] at (1,2.8) {8};
                \node[anchor=east,inner sep=1pt, blue] at (2,2.8) {1};
                \node[anchor=east,inner sep=1pt, blue] at (3,2.8) {0};
                \node[anchor=east,inner sep=1pt, blue] at (4,2.8) {3};

                \node[anchor=east,inner sep=1pt, blue] at (1,3.8) {7};
                \node[anchor=east,inner sep=1pt, blue] at (1.4, 3.2) {2};
                \node[anchor=east,inner sep=1pt, blue] at (2.4,3.2) {1};
                \node[anchor=east,inner sep=1pt, blue] at (4,3.8) {2};

                \node[anchor=east,inner sep=1pt, blue] at (1,4.8) {6};
                \node[anchor=east,inner sep=1pt, blue] at (2,4.8) {3};
                \node[anchor=east,inner sep=1pt, blue] at (3,4.8) {1};
                \node[anchor=east,inner sep=1pt, blue] at (4,4.8) {8};

                \node[anchor=east,inner sep=1pt, blue] at (1,5.8) {3};
                \node[anchor=east,inner sep=1pt, blue] at (2,5.8) {0};
                \node[anchor=east,inner sep=1pt, blue] at (2.4,5.2) {2};
                \node[anchor=east,inner sep=1pt, blue] at (4,5.8) {7};

                \node[anchor=east,inner sep=1pt, blue] at (1,6.8) {2};
                \node[anchor=east,inner sep=1pt, blue] at (1.4,6.2) {1};
                \node[anchor=east,inner sep=1pt, blue] at (3,6.8) {3};
                \node[anchor=east,inner sep=1pt, blue] at (4,6.8) {6};

                \node[anchor=east,inner sep=1pt, blue] at (1,7.8) {5};
                \node[anchor=east,inner sep=1pt, blue] at (1.4,7.2) {4};
                \node[anchor=east,inner sep=1pt, blue] at (2.4,7.2) {4};
                \node[anchor=east,inner sep=1pt, blue] at (4,7.8) {5};

            \end{scope}
        \end{tikzpicture}
    }
    \subfloat{
        \begin{tikzpicture}[xscale=0.7,yscale=0.7]
            \initcube\printcube
            \flipup\printcube
            \flipup\printcube
            \flipup\printcube
            \flipup\printcube
            \flipup\printcube
            \rollup\printcube
            \flipup\printcube
            \rollup\printcube
            \flipright\printcube
            \rolldown\printcube
            \flipdown\printcube
            \rolldown\printcube
            \flipdown\printcube
            \flipdown\printcube
            \flipdown\printcube
            \flipdown\printcube
            \rollright\printcube
            \flipup\printcube
            \rollup\printcube
            \flipup\printcube
            \rollup\printcube
            \flipup\printcube
            \flipup\printcube
            \flipup\printcube
            \flipright\printcube
            \flipdown\printcube
            \flipdown\printcube
            \flipdown\printcube
            \rolldown\printcube
            \flipdown\printcube
            \rolldown\printcube
            \flipdown\printcube
            \flipdown\printcube
            \rollleft\printcube
            \rollleft\printcube
            \begin{scope}[shift={(-.5,-.5)}]
                \draw[gray, very thin] (0, 0) grid (4, 9); 
                \draw[blue, very thick] (0, 0) rectangle (4, 9);
                \draw[blue, very thick] (1, 1) -- (3, 1);
                \draw[blue, very thick] (2, 3) -- (2, 5);   
                \draw[blue, very thick] (2, 6) -- (2, 8);   

                \node[anchor=east,inner sep=1pt, blue] at (1,0.8) {12};
                \node[anchor=east,inner sep=1pt, blue] at (2,0.8) {0};
                \node[anchor=east,inner sep=1pt, blue] at (3,0.8) {0};
                \node[anchor=east,inner sep=1pt, blue] at (4,0.8) {12};

                \node[anchor=east,inner sep=1pt, blue] at (1,1.8) {11};
                \node[anchor=east,inner sep=1pt, blue] at (1.4,1.2) {0};
                \node[anchor=east,inner sep=1pt, blue] at (2.4,1.2) {0};
                \node[anchor=east,inner sep=1pt, blue] at (4,1.8) {11};

                \node[anchor=east,inner sep=1pt, blue] at (1,2.8) {10};
                \node[anchor=east,inner sep=1pt, blue] at (2,2.8) {1};
                \node[anchor=east,inner sep=1pt, blue] at (3,2.8) {1};
                \node[anchor=east,inner sep=1pt, blue] at (4,2.8) {10};

                \node[anchor=east,inner sep=1pt, blue] at (1,3.8) {9};
                \node[anchor=east,inner sep=1pt, blue] at (1.4,3.2) {2};
                \node[anchor=east,inner sep=1pt, blue] at (3,3.8) {0};
                \node[anchor=east,inner sep=1pt, blue] at (4,3.8) {3};

                \node[anchor=east,inner sep=1pt, blue] at (1,4.8) {8};
                \node[anchor=east,inner sep=1pt, blue] at (2,4.8) {3};
                \node[anchor=east,inner sep=1pt, blue] at (3,4.8) {1};
                \node[anchor=east,inner sep=1pt, blue] at (3.4,4.2) {2};

                \node[anchor=east,inner sep=1pt, blue] at (1,5.8) {7};
                \node[anchor=east,inner sep=1pt, blue] at (1.4,5.2) {4};
                \node[anchor=east,inner sep=1pt, blue] at (2.4,5.2) {2};
                \node[anchor=east,inner sep=1pt, blue] at (4,5.8) {9};

                \node[anchor=east,inner sep=1pt, blue] at (1,6.8) {0};
                \node[anchor=east,inner sep=1pt, blue] at (1.4,6.2) {3};
                \node[anchor=east,inner sep=1pt, blue] at (3,6.8) {3};
                \node[anchor=east,inner sep=1pt, blue] at (4,6.8) {8};

                \node[anchor=east,inner sep=1pt, blue] at (1,7.8) {1};
                \node[anchor=east,inner sep=1pt, blue] at (2,7.8) {2};
                \node[anchor=east,inner sep=1pt, blue] at (2.4,7.2) {4};
                \node[anchor=east,inner sep=1pt, blue] at (4,7.8) {7};

                \node[anchor=east,inner sep=1pt, blue] at (1,8.8) {6};
                \node[anchor=east,inner sep=1pt, blue] at (2,8.8) {5};
                \node[anchor=east,inner sep=1pt, blue] at (3,8.8) {5};
                \node[anchor=east,inner sep=1pt, blue] at (4,8.8) {6};

            \end{scope}
        \end{tikzpicture}
    }
    \subfloat{
        \begin{tikzpicture}[xscale=0.7,yscale=0.7]
            \initcube\printcube
            \rollup\printcube
            \flipup\printcube
            \rollup\printcube
            \flipup\printcube
            \flipup\printcube
            \flipup\printcube
            \flipup\printcube
            \flipup\printcube
            \rollright\printcube
            \flipdown\printcube
            \flipdown\printcube
            \flipdown\printcube
            \rollright\printcube
            \rollup\printcube
            \flipup\printcube
            \rollup\printcube
            \flipright\printcube
            \rolldown\printcube
            \flipdown\printcube
            \rolldown\printcube
            \flipdown\printcube
            \flipdown\printcube
            \flipdown\printcube
            \flipdown\printcube
            \flipdown\printcube
            \rollleft\printcube
            \flipup\printcube
            \flipup\printcube
            \flipup\printcube
            \flipup\printcube
            \rollleft\printcube
            \flipdown\printcube
            \rolldown\printcube
            \flipdown\printcube
            \rolldown\printcube
            \begin{scope}[shift={(-.5,-.5)}]
                \draw[gray, very thin] (0, 0) grid (4, 9); 
                \draw[blue, very thick] (0, 0) rectangle (4, 9);
                \draw[blue, very thick] (2, 1) -- (2, 3);
                \draw[blue, very thick] (1, 5) -- (3, 5);   
                \draw[blue, very thick] (2, 6) -- (2, 8);   

                \node[anchor=east,inner sep=1pt, blue] at (1,0.8) {9};
                \node[anchor=east,inner sep=1pt, blue] at (2,0.8) {8};
                \node[anchor=east,inner sep=1pt, blue] at (3,0.8) {4};
                \node[anchor=east,inner sep=1pt, blue] at (4,0.8) {8};

                \node[anchor=east,inner sep=1pt, blue] at (1,1.8) {3};
                \node[anchor=east,inner sep=1pt, blue] at (2,1.8) {2};
                \node[anchor=east,inner sep=1pt, blue] at (2.4,1.2) {3};
                \node[anchor=east,inner sep=1pt, blue] at (3.4,1.2) {7};

                \node[anchor=east,inner sep=1pt, blue] at (1,2.8) {0};
                \node[anchor=east,inner sep=1pt, blue] at (1.4,2.2) {1};
                \node[anchor=east,inner sep=1pt, blue] at (3,2.8) {2};
                \node[anchor=east,inner sep=1pt, blue] at (4,2.8) {6};

                \node[anchor=east,inner sep=1pt, blue] at (1,3.8) {6};
                \node[anchor=east,inner sep=1pt, blue] at (1.4,3.2) {7};
                \node[anchor=east,inner sep=1pt, blue] at (2.4,3.2) {1};
                \node[anchor=east,inner sep=1pt, blue] at (3.4,3.2) {5};

                \node[anchor=east,inner sep=1pt, blue] at (1,4.8) {5};
                \node[anchor=east,inner sep=1pt, blue] at (2,4.8) {4};
                \node[anchor=east,inner sep=1pt, blue] at (3,4.8) {0};
                \node[anchor=east,inner sep=1pt, blue] at (4,4.8) {4};

                \node[anchor=east,inner sep=1pt, blue] at (1,5.8) {3};
                \node[anchor=east,inner sep=1pt, blue] at (2,5.8) {3};
                \node[anchor=east,inner sep=1pt, blue] at (3,5.8) {3};
                \node[anchor=east,inner sep=1pt, blue] at (3.4,5.2) {3};

                \node[anchor=east,inner sep=1pt, blue] at (1,6.8) {1};
                \node[anchor=east,inner sep=1pt, blue] at (2,6.8) {1};
                \node[anchor=east,inner sep=1pt, blue] at (3,6.8) {0};
                \node[anchor=east,inner sep=1pt, blue] at (4,6.8) {1};

                \node[anchor=east,inner sep=1pt, blue] at (1,7.8) {0};
                \node[anchor=east,inner sep=1pt, blue] at (2,7.8) {0};
                \node[anchor=east,inner sep=1pt, blue] at (2.4,7.2) {0};
                \node[anchor=east,inner sep=1pt, blue] at (4,7.8) {1};

                \node[anchor=east,inner sep=1pt, blue] at (1,8.8) {6};
                \node[anchor=east,inner sep=1pt, blue] at (2,8.8) {5};
                \node[anchor=east,inner sep=1pt, blue] at (3,8.8) {5};
                \node[anchor=east,inner sep=1pt, blue] at (4,8.8) {6};

            \end{scope}
        \end{tikzpicture}
    }
    \subfloat{

        \begin{tikzpicture}[xscale=0.7,yscale=0.7]
            \initcube\printcube
            \flipup\printcube
            \flipup\printcube
            \rollup\printcube
            \flipup\printcube
            \rollup\printcube
            \flipright\printcube
            \flipdown\printcube
            \flipdown\printcube
            \flipdown\printcube
            \flipdown\printcube
            \rollright\printcube
            \rollup\printcube
            \flipup\printcube
            \rollup\printcube
            \flipup\printcube
            \flipright\printcube
            \flipdown\printcube
            \rolldown\printcube
            \flipdown\printcube
            \rolldown\printcube
            \flipdown\printcube
            \rollleft\printcube
            \rollleft\printcube
            \begin{scope}[shift={(-.5,-.5)}]
                \draw[gray, very thin] (0, 0) grid (4, 6); 
                \draw[blue, very thick] (0, 0) rectangle (4, 6);
                \draw[blue, very thick] (1, 1) -- (3, 1);
                \draw[blue, very thick] (2, 2) -- (2, 4);   
                \draw[blue, very thick] (1, 3) -- (1, 5);   

                \node[anchor=east,inner sep=1pt, blue] at (1,0.8) {6};
                \node[anchor=east,inner sep=1pt, blue] at (2,0.8) {0};
                \node[anchor=east,inner sep=1pt, blue] at (3,0.8) {0};
                \node[anchor=east,inner sep=1pt, blue] at (4,0.8) {6};

                \node[anchor=east,inner sep=1pt, blue] at (1,1.8) {7};
                \node[anchor=east,inner sep=1pt, blue] at (1.4,1.2) {8};
                \node[anchor=east,inner sep=1pt, blue] at (2.4,1.2) {4};
                \node[anchor=east,inner sep=1pt, blue] at (4,1.8) {5};

                \node[anchor=east,inner sep=1pt, blue] at (1,2.8) {5};
                \node[anchor=east,inner sep=1pt, blue] at (2,2.8) {4};
                \node[anchor=east,inner sep=1pt, blue] at (3,2.8) {2};
                \node[anchor=east,inner sep=1pt, blue] at (4,2.8) {3};

                \node[anchor=east,inner sep=1pt, blue] at (1,3.8) {1};
                \node[anchor=east,inner sep=1pt, blue] at (1.4,3.2) {3};
                \node[anchor=east,inner sep=1pt, blue] at (2.4,3.2) {1};
                \node[anchor=east,inner sep=1pt, blue] at (4,3.8) {0};

                \node[anchor=east,inner sep=1pt, blue] at (1,4.8) {0};
                \node[anchor=east,inner sep=1pt, blue] at (2,4.8) {2};
                \node[anchor=east,inner sep=1pt, blue] at (3,4.8) {3};
                \node[anchor=east,inner sep=1pt, blue] at (4,4.8) {2};

                \node[anchor=east,inner sep=1pt, blue] at (1,5.8) {0};
                \node[anchor=east,inner sep=1pt, blue] at (1.4,5.2) {1};
                \node[anchor=east,inner sep=1pt, blue] at (2.4,5.2) {0};
                \node[anchor=east,inner sep=1pt, blue] at (4,5.8) {1};

            \end{scope}
        \end{tikzpicture}
    }\\
    \subfloat{
        \begin{tikzpicture}[xscale=0.7,yscale=0.7]
            \initcube\printcube
            \flipup\printcube
            \rollup\printcube
            \flipup\printcube
            \rollup\printcube
            \flipup\printcube
            \flipright\printcube
            \flipdown\printcube
            \flipdown\printcube
            \flipdown\printcube
            \flipdown\printcube
            \rollright\printcube
            \flipup\printcube
            \rollup\printcube
            \flipup\printcube
            \rollup\printcube
            \flipright\printcube
            \rolldown\printcube
            \flipdown\printcube
            \rolldown\printcube
            \flipdown\printcube
            \flipdown\printcube
            \rollleft\printcube
            \rollleft\printcube
            \begin{scope}[shift={(-.5,-.5)}]
                \draw[gray, very thin] (0, 0) grid (4, 6); 
                \draw[blue, very thick] (0, 0) rectangle (4, 6);
                \draw[blue, very thick] (1, 1) -- (3, 1);
                \draw[blue, very thick] (1, 2) -- (1, 4);   
                \draw[blue, very thick] (2, 3) -- (2, 5);   

                \node[anchor=east,inner sep=1pt, blue] at (1,0.8) {8};
                \node[anchor=east,inner sep=1pt, blue] at (2,0.8) {0};
                \node[anchor=east,inner sep=1pt, blue] at (3,0.8) {0};
                \node[anchor=east,inner sep=1pt, blue] at (4,0.8) {6};

                \node[anchor=east,inner sep=1pt, blue] at (1,1.8) {7};
                \node[anchor=east,inner sep=1pt, blue] at (1.4,1.2) {6};
                \node[anchor=east,inner sep=1pt, blue] at (2.4,1.2) {4};
                \node[anchor=east,inner sep=1pt, blue] at (4,1.8) {5};

                \node[anchor=east,inner sep=1pt, blue] at (1,2.8) {1};
                \node[anchor=east,inner sep=1pt, blue] at (2,2.8) {5};
                \node[anchor=east,inner sep=1pt, blue] at (3,2.8) {3};
                \node[anchor=east,inner sep=1pt, blue] at (4,2.8) {2};

                \node[anchor=east,inner sep=1pt, blue] at (1,3.8) {0};
                \node[anchor=east,inner sep=1pt, blue] at (1.4,3.2) {4};
                \node[anchor=east,inner sep=1pt, blue] at (3,3.8) {3};
                \node[anchor=east,inner sep=1pt, blue] at (4,3.8) {2};

                \node[anchor=east,inner sep=1pt, blue] at (1,4.8) {2};
                \node[anchor=east,inner sep=1pt, blue] at (2,4.8) {3};
                \node[anchor=east,inner sep=1pt, blue] at (3,4.8) {0};
                \node[anchor=east,inner sep=1pt, blue] at (4,4.8) {1};

                \node[anchor=east,inner sep=1pt, blue] at (1,5.8) {1};
                \node[anchor=east,inner sep=1pt, blue] at (1.4,5.2) {0};
                \node[anchor=east,inner sep=1pt, blue] at (2.4,5.2) {0};
                \node[anchor=east,inner sep=1pt, blue] at (4,5.8) {1};

            \end{scope}
        \end{tikzpicture}
    }
    \subfloat{
        \begin{tikzpicture}[xscale=0.7,yscale=0.7]
            \initcube\printcube
            \flipup\printcube
            \flipup\printcube
            \flipup\printcube
            \rollup\printcube
            \flipup\printcube
            \rollup\printcube
            \flipright\printcube
            \flipdown\printcube
            \flipdown\printcube
            \flipdown\printcube
            \flipdown\printcube
            \flipdown\printcube
            \rollright\printcube
            \flipup\printcube
            \rollup\printcube
            \flipup\printcube
            \rollup\printcube
            \flipup\printcube
            \flipright\printcube
            \flipdown\printcube
            \rolldown\printcube
            \flipdown\printcube
            \rolldown\printcube
            \flipdown\printcube
            \flipdown\printcube
            \rollleft\printcube
            \rollleft\printcube
            \begin{scope}[shift={(-.5,-.5)}]
                \draw[gray, very thin] (0, 0) grid (4, 7); 
                \draw[blue, very thick] (0, 0) rectangle (4, 7);
                \draw[blue, very thick] (1, 1) -- (3, 1);
                \draw[blue, very thick] (2, 3) -- (2, 5);   
                \draw[blue, very thick] (1, 4) -- (1, 6);   

                \node[anchor=east,inner sep=1pt, blue] at (1,0.8) {10};
                \node[anchor=east,inner sep=1pt, blue] at (2,0.8) {0};
                \node[anchor=east,inner sep=1pt, blue] at (3,0.8) {0};
                \node[anchor=east,inner sep=1pt, blue] at (4,0.8) {6};

                \node[anchor=east,inner sep=1pt, blue] at (1,1.8) {9};
                \node[anchor=east,inner sep=1pt, blue] at (1.4,1.2) {8};
                \node[anchor=east,inner sep=1pt, blue] at (2.4,1.2) {8};
                \node[anchor=east,inner sep=1pt, blue] at (4,1.8) {7};

                \node[anchor=east,inner sep=1pt, blue] at (1,2.8) {6};
                \node[anchor=east,inner sep=1pt, blue] at (2,2.8) {7};
                \node[anchor=east,inner sep=1pt, blue] at (3,2.8) {4};
                \node[anchor=east,inner sep=1pt, blue] at (4,2.8) {5};

                \node[anchor=east,inner sep=1pt, blue] at (1,3.8) {5};
                \node[anchor=east,inner sep=1pt, blue] at (1.4,3.2) {4};
                \node[anchor=east,inner sep=1pt, blue] at (3,3.8) {2};
                \node[anchor=east,inner sep=1pt, blue] at (4,3.8) {3};

                \node[anchor=east,inner sep=1pt, blue] at (1,4.8) {1};
                \node[anchor=east,inner sep=1pt, blue] at (2,4.8) {3};
                \node[anchor=east,inner sep=1pt, blue] at (3,4.8) {1};
                \node[anchor=east,inner sep=1pt, blue] at (3.4,4.2) {0};

                \node[anchor=east,inner sep=1pt, blue] at (1,5.8) {0};
                \node[anchor=east,inner sep=1pt, blue] at (1.4,5.2) {2};
                \node[anchor=east,inner sep=1pt, blue] at (2.4,5.2) {3};
                \node[anchor=east,inner sep=1pt, blue] at (4,5.8) {2};

                \node[anchor=east,inner sep=1pt, blue] at (1,6.8) {0};
                \node[anchor=east,inner sep=1pt, blue] at (2,6.8) {1};
                \node[anchor=east,inner sep=1pt, blue] at (3,6.8) {0};
                \node[anchor=east,inner sep=1pt, blue] at (4,6.8) {1};

            \end{scope}
        \end{tikzpicture}
    }
    \subfloat{
        \begin{tikzpicture}[xscale=0.7,yscale=0.7]
            \initcube\printcube
            \flipup\printcube
            \flipup\printcube
            \rollup\printcube
            \flipup\printcube
            \rollup\printcube
            \flipup\printcube
            \flipright\printcube
            \flipdown\printcube
            \flipdown\printcube
            \flipdown\printcube
            \flipdown\printcube
            \flipdown\printcube
            \rollright\printcube
            \flipup\printcube
            \flipup\printcube
            \rollup\printcube
            \flipup\printcube
            \rollup\printcube
            \flipright\printcube
            \rolldown\printcube
            \flipdown\printcube
            \rolldown\printcube
            \flipdown\printcube
            \flipdown\printcube
            \flipdown\printcube
            \rollleft\printcube
            \rollleft\printcube
            \begin{scope}[shift={(-.5,-.5)}]
                \draw[gray, very thin] (0, 0) grid (4, 7); 
                \draw[blue, very thick] (0, 0) rectangle (4, 7);
                \draw[blue, very thick] (1, 1) -- (3, 1);
                \draw[blue, very thick] (2, 4) -- (2, 6);   
                \draw[blue, very thick] (1, 3) -- (1, 5);   

                \node[anchor=east,inner sep=1pt, blue] at (1,0.8) {10};
                \node[anchor=east,inner sep=1pt, blue] at (2,0.8) {0};
                \node[anchor=east,inner sep=1pt, blue] at (3,0.8) {0};
                \node[anchor=east,inner sep=1pt, blue] at (4,0.8) {8};

                \node[anchor=east,inner sep=1pt, blue] at (1,1.8) {9};
                \node[anchor=east,inner sep=1pt, blue] at (1.4,1.2) {0};
                \node[anchor=east,inner sep=1pt, blue] at (2.4,1.2) {0};
                \node[anchor=east,inner sep=1pt, blue] at (4,1.8) {7};

                \node[anchor=east,inner sep=1pt, blue] at (1,2.8) {8};
                \node[anchor=east,inner sep=1pt, blue] at (2,2.8) {1};
                \node[anchor=east,inner sep=1pt, blue] at (3,2.8) {1};
                \node[anchor=east,inner sep=1pt, blue] at (4,2.8) {6};

                \node[anchor=east,inner sep=1pt, blue] at (1,3.8) {1};
                \node[anchor=east,inner sep=1pt, blue] at (1.4,3.2) {2};
                \node[anchor=east,inner sep=1pt, blue] at (2.4,3.2) {2};
                \node[anchor=east,inner sep=1pt, blue] at (3.4,3.2) {5};

                \node[anchor=east,inner sep=1pt, blue] at (1,4.8) {0};
                \node[anchor=east,inner sep=1pt, blue] at (2,4.8) {3};
                \node[anchor=east,inner sep=1pt, blue] at (3,4.8) {0};
                \node[anchor=east,inner sep=1pt, blue] at (4,4.8) {3};

                \node[anchor=east,inner sep=1pt, blue] at (1,5.8) {7};
                \node[anchor=east,inner sep=1pt, blue] at (1.4,5.2) {4};
                \node[anchor=east,inner sep=1pt, blue] at (2.4,5.2) {1};
                \node[anchor=east,inner sep=1pt, blue] at (3.4,5.2) {2};

                \node[anchor=east,inner sep=1pt, blue] at (1,6.8) {6};
                \node[anchor=east,inner sep=1pt, blue] at (2,6.8) {5};
                \node[anchor=east,inner sep=1pt, blue] at (3,6.8) {3};
                \node[anchor=east,inner sep=1pt, blue] at (4,6.8) {4};

            \end{scope}
        \end{tikzpicture}
    }
    \subfloat{
        \begin{tikzpicture}[xscale=0.7,yscale=0.7]
            \initcube\printcube
            \flipup\printcube
            \rollup\printcube
            \flipup\printcube
            \rollup\printcube
            \flipup\printcube
            \flipup\printcube
            \flipup\printcube
            \flipright\printcube
            \flipdown\printcube
            \flipdown\printcube
            \flipdown\printcube
            \flipdown\printcube
            \flipdown\printcube
            \flipdown\printcube
            \rollright\printcube
            \flipup\printcube
            \flipup\printcube
            \flipup\printcube
            \rollup\printcube
            \flipup\printcube
            \rollup\printcube
            \flipright\printcube
            \rolldown\printcube
            \flipdown\printcube
            \rolldown\printcube
            \flipdown\printcube
            \flipdown\printcube
            \flipdown\printcube
            \flipdown\printcube
            \rollleft\printcube
            \rollleft\printcube
            \begin{scope}[shift={(-.5,-.5)}]
                \draw[gray, very thin] (0, 0) grid (4, 8); 
                \draw[blue, very thick] (0, 0) rectangle (4, 8);
                \draw[blue, very thick] (1, 1) -- (3, 1);
                \draw[blue, very thick] (1, 2) -- (1, 4);   
                \draw[blue, very thick] (2, 5) -- (2, 7);   

                \node[anchor=east,inner sep=1pt, blue] at (1,0.8) {12};
                \node[anchor=east,inner sep=1pt, blue] at (2,0.8) {0};
                \node[anchor=east,inner sep=1pt, blue] at (3,0.8) {0};
                \node[anchor=east,inner sep=1pt, blue] at (4,0.8) {11};

                \node[anchor=east,inner sep=1pt, blue] at (1,1.8) {11};
                \node[anchor=east,inner sep=1pt, blue] at (1.4,1.2) {0};
                \node[anchor=east,inner sep=1pt, blue] at (2.4,1.2) {0};
                \node[anchor=east,inner sep=1pt, blue] at (4,1.8) {10};

                \node[anchor=east,inner sep=1pt, blue] at (1,2.8) {1};
                \node[anchor=east,inner sep=1pt, blue] at (2,2.8) {1};
                \node[anchor=east,inner sep=1pt, blue] at (3,2.8) {1};
                \node[anchor=east,inner sep=1pt, blue] at (4,2.8) {9};

                \node[anchor=east,inner sep=1pt, blue] at (1,3.8) {0};
                \node[anchor=east,inner sep=1pt, blue] at (1.4,3.2) {2};
                \node[anchor=east,inner sep=1pt, blue] at (2.4,3.2) {2};
                \node[anchor=east,inner sep=1pt, blue] at (4,3.8) {8};

                \node[anchor=east,inner sep=1pt, blue] at (1,4.8) {10};
                \node[anchor=east,inner sep=1pt, blue] at (2,4.8) {3};
                \node[anchor=east,inner sep=1pt, blue] at (3,4.8) {3};
                \node[anchor=east,inner sep=1pt, blue] at (4,4.8) {7};

                \node[anchor=east,inner sep=1pt, blue] at (1,5.8) {9};
                \node[anchor=east,inner sep=1pt, blue] at (1.4,5.2) {4};
                \node[anchor=east,inner sep=1pt, blue] at (3,5.8) {0};
                \node[anchor=east,inner sep=1pt, blue] at (4,5.8) {3};

                \node[anchor=east,inner sep=1pt, blue] at (1,6.8) {8};
                \node[anchor=east,inner sep=1pt, blue] at (2,6.8) {5};
                \node[anchor=east,inner sep=1pt, blue] at (2.4,6.2) {1};
                \node[anchor=east,inner sep=1pt, blue] at (3.4,6.2) {2};

                \node[anchor=east,inner sep=1pt, blue] at (1,7.8) {7};
                \node[anchor=east,inner sep=1pt, blue] at (1.4,7.2) {6};
                \node[anchor=east,inner sep=1pt, blue] at (2.4,7.2) {4};
                \node[anchor=east,inner sep=1pt, blue] at (3.4,7.2) {5};

            \end{scope}
        \end{tikzpicture}
    }\\
    \subfloat{
        \begin{tikzpicture}[xscale=0.7,yscale=0.7]

            \initcube\printcube
            \flipup\printcube
            \flipup\printcube
            \flipup\printcube
            \flipup\printcube
            \rollup\printcube
            \flipup\printcube
            \rollup\printcube
            \flipright\printcube
            \flipdown\printcube
            \flipdown\printcube
            \flipdown\printcube
            \flipdown\printcube
            \flipdown\printcube
            \flipdown\printcube
            \rollright\printcube
            \rollup\printcube
            \flipup\printcube
            \rollup\printcube
            \flipup\printcube
            \flipup\printcube
            \flipup\printcube
            \flipright\printcube
            \flipdown\printcube
            \flipdown\printcube
            \flipdown\printcube
            \rolldown\printcube
            \flipdown\printcube
            \rolldown\printcube
            \flipdown\printcube
            \rollleft\printcube
            \rollleft\printcube

            \begin{scope}[shift={(-.5,-.5)}]
                \draw[gray, very thin] (0, 0) grid (4, 8); 
                \draw[blue, very thick] (0, 0) rectangle (4, 8);
                \draw[blue, very thick] (1, 1) -- (3, 1);
                \draw[blue, very thick] (2, 2) -- (2, 4);   
                \draw[blue, very thick] (1, 5) -- (1, 7);   

                \node[anchor=east,inner sep=1pt, blue] at (1,0.8) {12};
                \node[anchor=east,inner sep=1pt, blue] at (2,0.8) {0};
                \node[anchor=east,inner sep=1pt, blue] at (3,0.8) {0};
                \node[anchor=east,inner sep=1pt, blue] at (4,0.8) {10};

                \node[anchor=east,inner sep=1pt, blue] at (1,1.8) {11};
                \node[anchor=east,inner sep=1pt, blue] at (1.4,1.2) {0};
                \node[anchor=east,inner sep=1pt, blue] at (2.4,1.2) {0};
                \node[anchor=east,inner sep=1pt, blue] at (3.5,1.2) {9};

                \node[anchor=east,inner sep=1pt, blue] at (1,2.8) {10};
                \node[anchor=east,inner sep=1pt, blue] at (2,2.8) {1};
                \node[anchor=east,inner sep=1pt, blue] at (3,2.8) {0};
                \node[anchor=east,inner sep=1pt, blue] at (4,2.8) {3};

                \node[anchor=east,inner sep=1pt, blue] at (1,3.8) {9};
                \node[anchor=east,inner sep=1pt, blue] at (1.4,3.2) {2};
                \node[anchor=east,inner sep=1pt, blue] at (2.4,3.2) {1};
                \node[anchor=east,inner sep=1pt, blue] at (4,3.8) {2};

                \node[anchor=east,inner sep=1pt, blue] at (1,4.8) {8};
                \node[anchor=east,inner sep=1pt, blue] at (2,4.8) {3};
                \node[anchor=east,inner sep=1pt, blue] at (3,4.8) {1};
                \node[anchor=east,inner sep=1pt, blue] at (4,4.8) {8};

                \node[anchor=east,inner sep=1pt, blue] at (1,5.8) {1};
                \node[anchor=east,inner sep=1pt, blue] at (1.4,5.2) {4};
                \node[anchor=east,inner sep=1pt, blue] at (2.4,5.2) {2};
                \node[anchor=east,inner sep=1pt, blue] at (3.4,5.2) {7};

                \node[anchor=east,inner sep=1pt, blue] at (1,6.8) {0};
                \node[anchor=east,inner sep=1pt, blue] at (2,6.8) {5};
                \node[anchor=east,inner sep=1pt, blue] at (3,6.8) {3};
                \node[anchor=east,inner sep=1pt, blue] at (4,6.8) {6};

                \node[anchor=east,inner sep=1pt, blue] at (1,7.8) {7};
                \node[anchor=east,inner sep=1pt, blue] at (1.4,7.2) {6};
                \node[anchor=east,inner sep=1pt, blue] at (2.4,7.2) {4};
                \node[anchor=east,inner sep=1pt, blue] at (3.4,7.2) {5};

            \end{scope}
        \end{tikzpicture}
    }
    \subfloat{
        \begin{tikzpicture}[xscale=0.7,yscale=0.7]
            \initcube\printcube
            \flipup\printcube
            \flipup\printcube
            \flipup\printcube
            \flipup\printcube
            \flipup\printcube
            \rollup\printcube
            \flipup\printcube
            \rollup\printcube
            \flipright\printcube
            \flipdown\printcube
            \flipdown\printcube
            \flipdown\printcube
            \flipdown\printcube
            \flipdown\printcube
            \flipdown\printcube
            \flipdown\printcube
            \rollright\printcube
            \flipup\printcube
            \rollup\printcube
            \flipup\printcube
            \rollup\printcube
            \flipup\printcube
            \flipup\printcube
            \flipup\printcube
            \rollright\printcube
            \flipdown\printcube
            \flipdown\printcube
            \flipdown\printcube
            \rolldown\printcube
            \flipdown\printcube
            \rolldown\printcube
            \flipdown\printcube
            \flipdown\printcube
            \rollleft\printcube
            \rollleft\printcube
            \begin{scope}[shift={(-.5,-.5)}]
                \draw[gray, very thin] (0, 0) grid (4, 9); 
                \draw[blue, very thick] (0, 0) rectangle (4, 9);
                \draw[blue, very thick] (1, 1) -- (3, 1);
                \draw[blue, very thick] (2, 3) -- (2, 5);   
                \draw[blue, very thick] (1, 6) -- (1, 8);   

                \node[anchor=east,inner sep=1pt, blue] at (1,0.8) {14};
                \node[anchor=east,inner sep=1pt, blue] at (2,0.8) {0};
                \node[anchor=east,inner sep=1pt, blue] at (3,0.8) {0};
                \node[anchor=east,inner sep=1pt, blue] at (4,0.8) {12};

                \node[anchor=east,inner sep=1pt, blue] at (1,1.8) {13};
                \node[anchor=east,inner sep=1pt, blue] at (1.4,1.2) {0};
                \node[anchor=east,inner sep=1pt, blue] at (2.4,1.2) {0};
                \node[anchor=east,inner sep=1pt, blue] at (3.5,1.2) {11};

                \node[anchor=east,inner sep=1pt, blue] at (1,2.8) {12};
                \node[anchor=east,inner sep=1pt, blue] at (2,2.8) {1};
                \node[anchor=east,inner sep=1pt, blue] at (3,2.8) {1};
                \node[anchor=east,inner sep=1pt, blue] at (4,2.8) {10};

                \node[anchor=east,inner sep=1pt, blue] at (1,3.8) {11};
                \node[anchor=east,inner sep=1pt, blue] at (1.4,3.2) {2};
                \node[anchor=east,inner sep=1pt, blue] at (3,3.8) {0};
                \node[anchor=east,inner sep=1pt, blue] at (4,3.8) {3};

                \node[anchor=east,inner sep=1pt, blue] at (1,4.8) {10};
                \node[anchor=east,inner sep=1pt, blue] at (2,4.8) {3};
                \node[anchor=east,inner sep=1pt, blue] at (3,4.8) {1};
                \node[anchor=east,inner sep=1pt, blue] at (4,4.8) {2};

                \node[anchor=east,inner sep=1pt, blue] at (1,5.8) {9};
                \node[anchor=east,inner sep=1pt, blue] at (1.4,5.2) {4};
                \node[anchor=east,inner sep=1pt, blue] at (2.4,5.2) {2};
                \node[anchor=east,inner sep=1pt, blue] at (3.4,5.2) {9};

                \node[anchor=east,inner sep=1pt, blue] at (1,6.8) {1};
                \node[anchor=east,inner sep=1pt, blue] at (2,6.8) {5};
                \node[anchor=east,inner sep=1pt, blue] at (3,6.8) {3};
                \node[anchor=east,inner sep=1pt, blue] at (4,6.8) {8};

                \node[anchor=east,inner sep=1pt, blue] at (1,7.8) {0};
                \node[anchor=east,inner sep=1pt, blue] at (1.4,7.2) {6};
                \node[anchor=east,inner sep=1pt, blue] at (2.4,7.2) {4};
                \node[anchor=east,inner sep=1pt, blue] at (3.4,7.2) {7};

                \node[anchor=east,inner sep=1pt, blue] at (1,8.8) {8};
                \node[anchor=east,inner sep=1pt, blue] at (2,8.8) {7};
                \node[anchor=east,inner sep=1pt, blue] at (3,8.8) {5};
                \node[anchor=east,inner sep=1pt, blue] at (4,8.8) {6};

            \end{scope}
        \end{tikzpicture}
    }
    \subfloat{
        \begin{tikzpicture}[xscale=0.7,yscale=0.7]
            \initcube\printcube
            \flipup\printcube
            \flipup\printcube
            \rollup\printcube
            \flipup\printcube
            \rollup\printcube
            \flipup\printcube
            \flipup\printcube
            \flipup\printcube
            \flipright\printcube
            \flipdown\printcube
            \flipdown\printcube
            \flipdown\printcube
            \flipdown\printcube
            \flipdown\printcube
            \flipdown\printcube
            \flipdown\printcube
            \rollright\printcube
            \flipup\printcube
            \flipup\printcube
            \flipup\printcube
            \flipup\printcube
            \rollup\printcube
            \flipup\printcube
            \rollup\printcube
            \flipright\printcube
            \rolldown\printcube
            \flipdown\printcube
            \rolldown\printcube
            \flipdown\printcube
            \flipdown\printcube
            \flipdown\printcube
            \flipdown\printcube
            \flipdown\printcube
            \rollleft\printcube
            \rollleft\printcube
            \begin{scope}[shift={(-.5,-.5)}]
                \draw[gray, very thin] (0, 0) grid (4, 9); 
                \draw[blue, very thick] (0, 0) rectangle (4, 9);
                \draw[blue, very thick] (1, 1) -- (3, 1);
                \draw[blue, very thick] (1, 3) -- (1, 5);   
                \draw[blue, very thick] (2, 6) -- (2, 8);   

                \node[anchor=east,inner sep=1pt, blue] at (1,0.8) {14};
                \node[anchor=east,inner sep=1pt, blue] at (2,0.8) {0};
                \node[anchor=east,inner sep=1pt, blue] at (3,0.8) {0};
                \node[anchor=east,inner sep=1pt, blue] at (4,0.8) {12};

                \node[anchor=east,inner sep=1pt, blue] at (1,1.8) {13};
                \node[anchor=east,inner sep=1pt, blue] at (1.4,1.2) {0};
                \node[anchor=east,inner sep=1pt, blue] at (2.4,1.2) {0};
                \node[anchor=east,inner sep=1pt, blue] at (3.5,1.2) {11};

                \node[anchor=east,inner sep=1pt, blue] at (1,2.8) {12};
                \node[anchor=east,inner sep=1pt, blue] at (2,2.8) {1};
                \node[anchor=east,inner sep=1pt, blue] at (3,2.8) {1};
                \node[anchor=east,inner sep=1pt, blue] at (4,2.8) {10};

                \node[anchor=east,inner sep=1pt, blue] at (1,3.8) {1};
                \node[anchor=east,inner sep=1pt, blue] at (1.4,3.2) {2};
                \node[anchor=east,inner sep=1pt, blue] at (2.4,3.2) {2};
                \node[anchor=east,inner sep=1pt, blue] at (4,3.8) {9};

                \node[anchor=east,inner sep=1pt, blue] at (1,4.8) {0};
                \node[anchor=east,inner sep=1pt, blue] at (2,4.8) {3};
                \node[anchor=east,inner sep=1pt, blue] at (3,4.8) {3};
                \node[anchor=east,inner sep=1pt, blue] at (4,4.8) {8};

                \node[anchor=east,inner sep=1pt, blue] at (1,5.8) {11};
                \node[anchor=east,inner sep=1pt, blue] at (1.4,5.2) {4};
                \node[anchor=east,inner sep=1pt, blue] at (2.4,5.2) {4};
                \node[anchor=east,inner sep=1pt, blue] at (3.4,5.2) {7};

                \node[anchor=east,inner sep=1pt, blue] at (1,6.8) {10};
                \node[anchor=east,inner sep=1pt, blue] at (2,6.8) {5};
                \node[anchor=east,inner sep=1pt, blue] at (3,6.8) {0};
                \node[anchor=east,inner sep=1pt, blue] at (4,6.8) {3};

                \node[anchor=east,inner sep=1pt, blue] at (1,7.8) {9};
                \node[anchor=east,inner sep=1pt, blue] at (1.4,7.2) {6};
                \node[anchor=east,inner sep=1pt, blue] at (2.4,7.2) {1};
                \node[anchor=east,inner sep=1pt, blue] at (3.4,7.2) {2};

                \node[anchor=east,inner sep=1pt, blue] at (1,8.8) {8};
                \node[anchor=east,inner sep=1pt, blue] at (2,8.8) {7};
                \node[anchor=east,inner sep=1pt, blue] at (3,8.8) {5};
                \node[anchor=east,inner sep=1pt, blue] at (4,8.8) {6};

            \end{scope}
        \end{tikzpicture}
    }
    \subfloat{
        \begin{tikzpicture}[xscale=0.7,yscale=0.7]
            \initcube\printcube
            \rollup\printcube
            \flipup\printcube
            \rollup\printcube
            \flipup\printcube
            \flipup\printcube
            \flipup\printcube
            \flipup\printcube
            \flipup\printcube
            \flipright\printcube
            \flipdown\printcube
            \flipdown\printcube
            \flipdown\printcube
            \rollright\printcube
            \rollup\printcube
            \flipup\printcube
            \rollup\printcube
            \flipright\printcube
            \rolldown\printcube
            \flipdown\printcube
            \rolldown\printcube
            \flipdown\printcube
            \flipdown\printcube
            \flipdown\printcube
            \flipdown\printcube
            \flipdown\printcube
            \rollleft\printcube
            \flipup\printcube
            \flipup\printcube
            \flipup\printcube
            \flipup\printcube
            \rollleft\printcube
            \flipdown\printcube
            \flipdown\printcube
            \flipdown\printcube
            \flipdown\printcube

            \begin{scope}[shift={(-.5,-.5)}]
                \draw[gray, very thin] (0, 0) grid (4, 9); 
                \draw[blue, very thick] (0, 0) rectangle (4, 9);
                \draw[blue, very thick] (1, 5) -- (3, 5);
                \draw[blue, very thick] (1, 1) -- (1, 3);   
                \draw[blue, very thick] (2, 6) -- (2, 8);   

                \node[anchor=east,inner sep=1pt, blue] at (1,0.8) {10};
                \node[anchor=east,inner sep=1pt, blue] at (2,0.8) {4};
                \node[anchor=east,inner sep=1pt, blue] at (3,0.8) {4};
                \node[anchor=east,inner sep=1pt, blue] at (4,0.8) {8};

                \node[anchor=east,inner sep=1pt, blue] at (1,1.8) {1};
                \node[anchor=east,inner sep=1pt, blue] at (2,1.8) {8};
                \node[anchor=east,inner sep=1pt, blue] at (2.4,1.2) {3};
                \node[anchor=east,inner sep=1pt, blue] at (3.4,1.2) {7};

                \node[anchor=east,inner sep=1pt, blue] at (1,2.8) {0};
                \node[anchor=east,inner sep=1pt, blue] at (2,2.8) {2};
                \node[anchor=east,inner sep=1pt, blue] at (3,2.8) {2};
                \node[anchor=east,inner sep=1pt, blue] at (4,2.8) {6};

                \node[anchor=east,inner sep=1pt, blue] at (1,3.8) {9};
                \node[anchor=east,inner sep=1pt, blue] at (1.4,3.2) {1};
                \node[anchor=east,inner sep=1pt, blue] at (2.4,3.2) {1};
                \node[anchor=east,inner sep=1pt, blue] at (3.4,3.2) {5};

                \node[anchor=east,inner sep=1pt, blue] at (1,4.8) {8};
                \node[anchor=east,inner sep=1pt, blue] at (2,4.8) {0};
                \node[anchor=east,inner sep=1pt, blue] at (3,4.8) {0};
                \node[anchor=east,inner sep=1pt, blue] at (4,4.8) {4};

                \node[anchor=east,inner sep=1pt, blue] at (1,5.8) {7};
                \node[anchor=east,inner sep=1pt, blue] at (1.4,5.2) {0};
                \node[anchor=east,inner sep=1pt, blue] at (2.4,5.2) {0};
                \node[anchor=east,inner sep=1pt, blue] at (3.4,5.2) {3};

                \node[anchor=east,inner sep=1pt, blue] at (1,6.8) {6};
                \node[anchor=east,inner sep=1pt, blue] at (2,6.8) {1};
                \node[anchor=east,inner sep=1pt, blue] at (3,6.8) {0};
                \node[anchor=east,inner sep=1pt, blue] at (4,6.8) {3};

                \node[anchor=east,inner sep=1pt, blue] at (1,7.8) {5};
                \node[anchor=east,inner sep=1pt, blue] at (1.4,7.2) {2};
                \node[anchor=east,inner sep=1pt, blue] at (2.4,7.2) {1};
                \node[anchor=east,inner sep=1pt, blue] at (4,7.8) {2};

                \node[anchor=east,inner sep=1pt, blue] at (1,8.8) {4};
                \node[anchor=east,inner sep=1pt, blue] at (2,8.8) {3};
                \node[anchor=east,inner sep=1pt, blue] at (3,8.8) {1};
                \node[anchor=east,inner sep=1pt, blue] at (4,8.8) {2};

            \end{scope}
        \end{tikzpicture}
    }
    \subfloat{
        \begin{tikzpicture}[xscale=0.7,yscale=0.7]
            \initcube\printcube
            \flipup\printcube
            \flipup\printcube
            \flipup\printcube
            \flipup\printcube
            \flipup\printcube
            \rollup\printcube
            \flipup\printcube
            \rollup\printcube
            \rollright\printcube
            \flipdown\printcube
            \flipdown\printcube
            \flipdown\printcube
            \rollright\printcube
            \flipup\printcube
            \flipup\printcube
            \flipup\printcube
            \rollright\printcube
            \flipdown\printcube
            \flipdown\printcube
            \flipdown\printcube
            \flipdown\printcube
            \flipdown\printcube
            \rolldown\printcube
            \flipdown\printcube
            \rolldown\printcube
            \flipleft\printcube
            \rollup\printcube
            \flipup\printcube
            \rollup\printcube
            \flipup\printcube
            \rollleft\printcube
            \flipdown\printcube
            \flipdown\printcube
            \flipdown\printcube
            \flipdown\printcube
            \begin{scope}[shift={(-.5,-.5)}]
                \draw[gray, very thin] (0, 0) grid (4, 9); 
                \draw[blue, very thick] (0, 0) rectangle (4, 9);
                \draw[blue, very thick] (1, 5) -- (3, 5);
                \draw[blue, very thick] (2, 1) -- (2, 3);   
                \draw[blue, very thick] (1, 6) -- (1, 8);   

                \node[anchor=east,inner sep=1pt, blue] at (1,0.8) {5};
                \node[anchor=east,inner sep=1pt, blue] at (2,0.8) {4};
                \node[anchor=east,inner sep=1pt, blue] at (3,0.8) {2};
                \node[anchor=east,inner sep=1pt, blue] at (4,0.8) {3};

                \node[anchor=east,inner sep=1pt, blue] at (1,1.8) {6};
                \node[anchor=east,inner sep=1pt, blue] at (2,1.8) {3};
                \node[anchor=east,inner sep=1pt, blue] at (3,1.8) {1};
                \node[anchor=east,inner sep=1pt, blue] at (4,1.8) {2};

                \node[anchor=east,inner sep=1pt, blue] at (1,2.8) {7};
                \node[anchor=east,inner sep=1pt, blue] at (2,2.8) {2};
                \node[anchor=east,inner sep=1pt, blue] at (2.4,2.2) {0};
                \node[anchor=east,inner sep=1pt, blue] at (3.4,2.2) {3};

                \node[anchor=east,inner sep=1pt, blue] at (1,3.8) {8};
                \node[anchor=east,inner sep=1pt, blue] at (1.4,3.2) {1};
                \node[anchor=east,inner sep=1pt, blue] at (2.4,3.2) {1};
                \node[anchor=east,inner sep=1pt, blue] at (3.4,3.2) {4};

                \node[anchor=east,inner sep=1pt, blue] at (1,4.8) {9};
                \node[anchor=east,inner sep=1pt, blue] at (2,4.8) {0};
                \node[anchor=east,inner sep=1pt, blue] at (3,4.8) {0};
                \node[anchor=east,inner sep=1pt, blue] at (4,4.8) {5};

                \node[anchor=east,inner sep=1pt, blue] at (1,5.8) {10};
                \node[anchor=east,inner sep=1pt, blue] at (1.4,5.2) {0};
                \node[anchor=east,inner sep=1pt, blue] at (2.4,5.2) {0};
                \node[anchor=east,inner sep=1pt, blue] at (3.4,5.2) {6};

                \node[anchor=east,inner sep=1pt, blue] at (1,6.8) {0};
                \node[anchor=east,inner sep=1pt, blue] at (2,6.8) {1};
                \node[anchor=east,inner sep=1pt, blue] at (3,6.8) {1};
                \node[anchor=east,inner sep=1pt, blue] at (4,6.8) {7};

                \node[anchor=east,inner sep=1pt, blue] at (1,7.8) {1};
                \node[anchor=east,inner sep=1pt, blue] at (1.4,7.2) {2};
                \node[anchor=east,inner sep=1pt, blue] at (2.4,7.2) {2};
                \node[anchor=east,inner sep=1pt, blue] at (4,7.8) {8};

                \node[anchor=east,inner sep=1pt, blue] at (1,8.8) {11};
                \node[anchor=east,inner sep=1pt, blue] at (2,8.8) {3};
                \node[anchor=east,inner sep=1pt, blue] at (3,8.8) {3};
                \node[anchor=east,inner sep=1pt, blue] at (4,8.8) {9};

            \end{scope}
        \end{tikzpicture}
    }
    \caption{Foldings of all minimal polyominoes with 1 horizontal and 2 vertical slits}
    \label{fig:slitsHorzontal}
\end{figure}

%% file: fourslits.tex
\begin{figure}
    \centering
    \subfloat{
        \begin{tikzpicture}[xscale=0.7,yscale=0.7]
            \initcube\printcube
            \rollup\printcube
            \rollup\printcube
            \rollup\printcube
            \rollup\printcube
            \flipright\printcube
            \flipdown\printcube
            \rolldown\printcube
            \flipdown\printcube
            \rolldown\printcube
            \rollright\printcube
            \rollup\printcube
            \flipup\printcube
            \rollup\printcube
            \flipup\printcube
            \flipright\printcube
            \rolldown\printcube
            \rolldown\printcube
            \rolldown\printcube
            \rolldown\printcube

            \begin{scope}[shift={(-.5,-.5)}]
                \draw[gray, very thin] (0, 0) grid (4, 5); 
                \draw[blue, very thick] (0, 0) rectangle (4, 5);
                \draw[blue, very thick] (1, 2) -- (1, 4);
                \draw[blue, very thick] (3, 2) -- (3, 4);   
                \draw[blue, very thick] (2, 1) -- (2, 3);   

                \node[anchor=east,inner sep=1pt, blue] at (1,0.8) {2};
                \node[anchor=east,inner sep=1pt, blue] at (2,0.8) {3};
                \node[anchor=east,inner sep=1pt, blue] at (3,0.8) {1};
                \node[anchor=east,inner sep=1pt, blue] at (4,0.8) {3};

                \node[anchor=east,inner sep=1pt, blue] at (1,1.8) {0};
                \node[anchor=east,inner sep=1pt, blue] at (2,1.8) {1};
                \node[anchor=east,inner sep=1pt, blue] at (3,1.8) {3};
                \node[anchor=east,inner sep=1pt, blue] at (4,1.8) {5};

                \node[anchor=east,inner sep=1pt, blue] at (0.4,2.2) {0};
                \node[anchor=east,inner sep=1pt, blue] at (1.4,2.2) {2};
                \node[anchor=east,inner sep=1pt, blue] at (2.4,2.2) {4};
                \node[anchor=east,inner sep=1pt, blue] at (3.4,2.2) {0};

                \node[anchor=east,inner sep=1pt, blue] at (1,3.8) {0};
                \node[anchor=east,inner sep=1pt, blue] at (1.4,3.2) {4};
                \node[anchor=east,inner sep=1pt, blue] at (2.4,3.2) {2};
                \node[anchor=east,inner sep=1pt, blue] at (4,3.8) {1};

                \node[anchor=east,inner sep=1pt, blue] at (1,4.8) {0};
                \node[anchor=east,inner sep=1pt, blue] at (2,4.8) {1};
                \node[anchor=east,inner sep=1pt, blue] at (3,4.8) {0};
                \node[anchor=east,inner sep=1pt, blue] at (4,4.8) {4};

            \end{scope}
        \end{tikzpicture}
    }
    \subfloat{
        \begin{tikzpicture}[xscale=0.7,yscale=0.7]
            \initcube\printcube
            \flipup\printcube
            \rollup\printcube
            \rollup\printcube
            \rollup\printcube
            \rollup\printcube
            \flipright\printcube
            \flipdown\printcube
            \rolldown\printcube
            \flipdown\printcube
            \rolldown\printcube
            \flipdown\printcube
            \rollright\printcube
            \flipup\printcube
            \rollup\printcube
            \flipup\printcube
            \rollup\printcube
            \flipup\printcube
            \flipright\printcube
            \flipdown\printcube
            \rolldown\printcube
            \rolldown\printcube
            \rolldown\printcube
            \rolldown\printcube
            \begin{scope}[shift={(-.5,-.5)}]
                \draw[gray, very thin] (0, 0) grid (4, 6); 
                \draw[blue, very thick] (0, 0) rectangle (4, 6);
                \draw[blue, very thick] (1, 3) -- (1, 5);
                \draw[blue, very thick] (3, 1) -- (3, 3);   
                \draw[blue, very thick] (2, 2) -- (2, 4);   

                \node[anchor=east,inner sep=1pt, blue] at (1,0.8) {5};
                \node[anchor=east,inner sep=1pt, blue] at (2,0.8) {6};
                \node[anchor=east,inner sep=1pt, blue] at (3,0.8) {5};
                \node[anchor=east,inner sep=1pt, blue] at (4,0.8) {6};

                \node[anchor=east,inner sep=1pt, blue] at (1,1.8) {4};
                \node[anchor=east,inner sep=1pt, blue] at (1.4,1.2) {3};
                \node[anchor=east,inner sep=1pt, blue] at (2.4,1.2) {4};
                \node[anchor=east,inner sep=1pt, blue] at (4,1.8) {1};

                \node[anchor=east,inner sep=1pt, blue] at (1,2.8) {2};
                \node[anchor=east,inner sep=1pt, blue] at (1.4,2.2) {1};
                \node[anchor=east,inner sep=1pt, blue] at (3,2.8) {5};
                \node[anchor=east,inner sep=1pt, blue] at (3.4,2.2) {0};

                \node[anchor=east,inner sep=1pt, blue] at (1,3.8) {0};
                \node[anchor=east,inner sep=1pt, blue] at (2,3.8) {0};
                \node[anchor=east,inner sep=1pt, blue] at (2.4,3.2) {4};
                \node[anchor=east,inner sep=1pt, blue] at (4,3.8) {3};

                \node[anchor=east,inner sep=1pt, blue] at (1,4.8) {0};
                \node[anchor=east,inner sep=1pt, blue] at (2,4.8) {2};
                \node[anchor=east,inner sep=1pt, blue] at (3,4.8) {3};
                \node[anchor=east,inner sep=1pt, blue] at (4,4.8) {2};

                \node[anchor=east,inner sep=1pt, blue] at (1,5.8) {0};
                \node[anchor=east,inner sep=1pt, blue] at (1.4,5.2) {1};
                \node[anchor=east,inner sep=1pt, blue] at (2.4,5.2) {0};
                \node[anchor=east,inner sep=1pt, blue] at (4,5.8) {1};
            \end{scope}
        \end{tikzpicture}
    }
    \subfloat{
        \begin{tikzpicture}[xscale=0.7,yscale=0.7]
            \initcube\printcube
            \rollup\printcube
            \flipup\printcube
            \rollup\printcube
            \rollup\printcube
            \rollup\printcube
            \rollup\printcube
            \rollup\printcube
            \flipright\printcube
            \rolldown\printcube
            \flipdown\printcube
            \rolldown\printcube
            \flipdown\printcube
            \rolldown\printcube
            \flipdown\printcube
            \rolldown\printcube
            \rollright\printcube
            \rollup\printcube
            \flipup\printcube
            \rollup\printcube
            \flipup\printcube
            \flipup\printcube
            \flipup\printcube
            \flipup\printcube
            \flipright\printcube
            \flipdown\printcube
            \flipdown\printcube
            \flipdown\printcube
            \rolldown\printcube
            \rolldown\printcube
            \rolldown\printcube
            \flipdown\printcube

            \begin{scope}[shift={(-.5,-.5)}]
                \draw[gray, very thin] (0, 0) grid (4, 8); 
                \draw[blue, very thick] (0, 0) rectangle (4, 8);
                \draw[blue, very thick] (2, 1) -- (2, 3);
                \draw[blue, very thick] (3, 2) -- (3, 4);   
                \draw[blue, very thick] (1, 4) -- (1, 6);   
                \draw[blue, very thick] (2, 5) -- (2, 7);   

                \node[anchor=east,inner sep=1pt, blue] at (1,0.8) {2};
                \node[anchor=east,inner sep=1pt, blue] at (2,0.8) {3};
                \node[anchor=east,inner sep=1pt, blue] at (3,0.8) {4};
                \node[anchor=east,inner sep=1pt, blue] at (4,0.8) {0};

                \node[anchor=east,inner sep=1pt, blue] at (1,1.8) {5};
                \node[anchor=east,inner sep=1pt, blue] at (2,1.8) {6};
                \node[anchor=east,inner sep=1pt, blue] at (3,1.8) {1};
                \node[anchor=east,inner sep=1pt, blue] at (4,1.8) {10};

                \node[anchor=east,inner sep=1pt, blue] at (1,2.8) {4};
                \node[anchor=east,inner sep=1pt, blue] at (1.4,2.2) {7};
                \node[anchor=east,inner sep=1pt, blue] at (3,2.8) {2};
                \node[anchor=east,inner sep=1pt, blue] at (3.4,2.2) {0};

                \node[anchor=east,inner sep=1pt, blue] at (1,3.8) {1};
                \node[anchor=east,inner sep=1pt, blue] at (1.4,3.2) {4};
                \node[anchor=east,inner sep=1pt, blue] at (2.4,3.2) {2};
                \node[anchor=east,inner sep=1pt, blue] at (4,3.8) {0};

                \node[anchor=east,inner sep=1pt, blue] at (1,4.8) {1};
                \node[anchor=east,inner sep=1pt, blue] at (2,4.8) {6};
                \node[anchor=east,inner sep=1pt, blue] at (3,4.8) {5};
                \node[anchor=east,inner sep=1pt, blue] at (4,4.8) {6};

                \node[anchor=east,inner sep=1pt, blue] at (1,5.8) {0};
                \node[anchor=east,inner sep=1pt, blue] at (1.4,5.2) {9};
                \node[anchor=east,inner sep=1pt, blue] at (2.4,5.2) {8};
                \node[anchor=east,inner sep=1pt, blue] at (4,5.8) {7};

                \node[anchor=east,inner sep=1pt, blue] at (1,6.8) {3};
                \node[anchor=east,inner sep=1pt, blue] at (1.4,6.2) {8};
                \node[anchor=east,inner sep=1pt, blue] at (2.4,6.2) {9};
                \node[anchor=east,inner sep=1pt, blue] at (4,6.8) {10};

                \node[anchor=east,inner sep=1pt, blue] at (1,7.8) {0};
                \node[anchor=east,inner sep=1pt, blue] at (1.4,7.2) {5};
                \node[anchor=east,inner sep=1pt, blue] at (2.4,7.2) {4};
                \node[anchor=east,inner sep=1pt, blue] at (4,7.8) {3};

            \end{scope}
        \end{tikzpicture}
    }\\
    \subfloat{
        \begin{tikzpicture}[xscale=0.7,yscale=0.7]
            \initcube\printcube
            \flipup\printcube
            \flipup\printcube
            \flipup\printcube
            \flipup\printcube
            \rollup\printcube
            \rollup\printcube
            \rollup\printcube
            \rollup\printcube
            \flipright\printcube
            \rolldown\printcube
            \flipdown\printcube
            \rolldown\printcube
            \flipdown\printcube
            \flipdown\printcube
            \flipdown\printcube
            \flipdown\printcube
            \flipdown\printcube
            \rollright\printcube
            \flipup\printcube
            \rollup\printcube
            \flipup\printcube
            \rollup\printcube
            \flipup\printcube
            \flipup\printcube
            \flipup\printcube
            \flipup\printcube
            \flipright\printcube
            \flipdown\printcube
            \flipdown\printcube
            \flipdown\printcube
            \flipdown\printcube
            \rolldown\printcube
            \rolldown\printcube
            \rolldown\printcube
            \rolldown\printcube
            \begin{scope}[shift={(-.5,-.5)}]
                \draw[gray, very thin] (0, 0) grid (4, 9); 
                \draw[blue, very thick] (0, 0) rectangle (4, 9);
                \draw[blue, very thick] (3, 1) -- (3, 3);
                \draw[blue, very thick] (2, 2) -- (2, 4);   
                \draw[blue, very thick] (1, 5) -- (1, 7);   
                \draw[blue, very thick] (2, 6) -- (2, 8);   

                \node[anchor=east,inner sep=1pt, blue] at (1,0.8) {2};
                \node[anchor=east,inner sep=1pt, blue] at (2,0.8) {9};
                \node[anchor=east,inner sep=1pt, blue] at (3,0.8) {3};
                \node[anchor=east,inner sep=1pt, blue] at (4,0.8) {0};

                \node[anchor=east,inner sep=1pt, blue] at (1,1.8) {1};
                \node[anchor=east,inner sep=1pt, blue] at (1.4,1.2) {10};
                \node[anchor=east,inner sep=1pt, blue] at (2.4,1.2) {4};
                \node[anchor=east,inner sep=1pt, blue] at (4,1.8) {1};

                \node[anchor=east,inner sep=1pt, blue] at (1,2.8) {5};
                \node[anchor=east,inner sep=1pt, blue] at (2,2.8) {6};
                \node[anchor=east,inner sep=1pt, blue] at (3,2.8) {5};
                \node[anchor=east,inner sep=1pt, blue] at (3.4,2.2) {0};

                \node[anchor=east,inner sep=1pt, blue] at (1,3.8) {4};
                \node[anchor=east,inner sep=1pt, blue] at (1.4,3.2) {7};
                \node[anchor=east,inner sep=1pt, blue] at (2.4, 3.2) {4};
                \node[anchor=east,inner sep=1pt, blue] at (4,3.8) {3};

                \node[anchor=east,inner sep=1pt, blue] at (1,4.8) {3};
                \node[anchor=east,inner sep=1pt, blue] at (2,4.8) {8};
                \node[anchor=east,inner sep=1pt, blue] at (3,4.8) {2};
                \node[anchor=east,inner sep=1pt, blue] at (4,4.8) {1};

                \node[anchor=east,inner sep=1pt, blue] at (1,5.8) {0};
                \node[anchor=east,inner sep=1pt, blue] at (1.4,5.2) {12};
                \node[anchor=east,inner sep=1pt, blue] at (2.4,5.2) {7};
                \node[anchor=east,inner sep=1pt, blue] at (4,5.8) {6};

                \node[anchor=east,inner sep=1pt, blue] at (0.4,6.2) {0};
                \node[anchor=east,inner sep=1pt, blue] at (1.4,6.2) {2};
                \node[anchor=east,inner sep=1pt, blue] at (3,6.8) {10};
                \node[anchor=east,inner sep=1pt, blue] at (4,6.8) {9};

                \node[anchor=east,inner sep=1pt, blue] at (1,7.8) {0};
                \node[anchor=east,inner sep=1pt, blue] at (1.4,7.2) {1};
                \node[anchor=east,inner sep=1pt, blue] at (2.4,7.2) {11};
                \node[anchor=east,inner sep=1pt, blue] at (4,7.8) {12};

                \node[anchor=east,inner sep=1pt, blue] at (1,8.8) {0};
                \node[anchor=east,inner sep=1pt, blue] at (2,8.8) {11};
                \node[anchor=east,inner sep=1pt, blue] at (3,8.8) {6};
                \node[anchor=east,inner sep=1pt, blue] at (4,8.8) {5};

            \end{scope}
        \end{tikzpicture}
    }
    \subfloat{
        \begin{tikzpicture}[xscale=0.7,yscale=0.7]
            \initcube\printcube
            \flipup\printcube
            \flipup\printcube
            \flipup\printcube
            \flipup\printcube
            \flipup\printcube
            \rollup\printcube
            \rollup\printcube
            \rollup\printcube
            \rollup\printcube
            \flipright\printcube
            \flipdown\printcube
            \rolldown\printcube
            \flipdown\printcube
            \rolldown\printcube
            \flipdown\printcube
            \flipdown\printcube
            \flipdown\printcube
            \flipdown\printcube
            \flipdown\printcube
            \rollright\printcube
            \flipup\printcube
            \rollup\printcube
            \flipup\printcube
            \rollup\printcube
            \flipup\printcube
            \rollup\printcube
            \flipup\printcube
            \rollup\printcube
            \flipup\printcube
            \flipright\printcube
            \flipdown\printcube
            \rolldown\printcube
            \flipdown\printcube
            \rolldown\printcube
            \flipdown\printcube
            \rolldown\printcube
            \rolldown\printcube
            \rolldown\printcube
            \rolldown\printcube
            \begin{scope}[shift={(-.5,-.5)}]
                \draw[gray, very thin] (0, 0) grid (4, 10); 
                \draw[blue, very thick] (0, 0) rectangle (4, 10);
                \draw[blue, very thick] (3, 1) -- (3, 3);
                \draw[blue, very thick] (2, 2) -- (2, 4);   
                \draw[blue, very thick] (1, 7) -- (1, 9);   
                \draw[blue, very thick] (2, 6) -- (2, 8);   

                \node[anchor=east,inner sep=1pt, blue] at (1,0.8) {8};
                \node[anchor=east,inner sep=1pt, blue] at (2,0.8) {7};
                \node[anchor=east,inner sep=1pt, blue] at (3,0.8) {9};
                \node[anchor=east,inner sep=1pt, blue] at (4,0.8) {0};

                \node[anchor=east,inner sep=1pt, blue] at (1,1.8) {9};
                \node[anchor=east,inner sep=1pt, blue] at (1.4,1.2) {10};
                \node[anchor=east,inner sep=1pt, blue] at (2.4,1.2) {10};
                \node[anchor=east,inner sep=1pt, blue] at (4,1.8) {0};

                \node[anchor=east,inner sep=1pt, blue] at (1,2.8) {12};
                \node[anchor=east,inner sep=1pt, blue] at (2,2.8) {11};
                \node[anchor=east,inner sep=1pt, blue] at (3,2.8) {6};
                \node[anchor=east,inner sep=1pt, blue] at (3.4,2.2) {0};

                \node[anchor=east,inner sep=1pt, blue] at (1,3.8) {13};
                \node[anchor=east,inner sep=1pt, blue] at (1.4,3.2) {14};
                \node[anchor=east,inner sep=1pt, blue] at (2.4,3.2) {5};
                \node[anchor=east,inner sep=1pt, blue] at (4,3.8) {0};

                \node[anchor=east,inner sep=1pt, blue] at (1,4.8) {6};
                \node[anchor=east,inner sep=1pt, blue] at (1.4,4.2) {5};
                \node[anchor=east,inner sep=1pt, blue] at (2.4,4.2) {8};
                \node[anchor=east,inner sep=1pt, blue] at (4,4.8) {1};

                \node[anchor=east,inner sep=1pt, blue] at (1,5.8) {1};
                \node[anchor=east,inner sep=1pt, blue] at (1.4,5.2) {2};
                \node[anchor=east,inner sep=1pt, blue] at (2.4,5.2) {5};
                \node[anchor=east,inner sep=1pt, blue] at (3.4,5.2) {4};

                \node[anchor=east,inner sep=1pt, blue] at (1,6.8) {7};
                \node[anchor=east,inner sep=1pt, blue] at (2,6.8) {8};
                \node[anchor=east,inner sep=1pt, blue] at (3,6.8) {3};
                \node[anchor=east,inner sep=1pt, blue] at (4,6.8) {2};

                \node[anchor=east,inner sep=1pt, blue] at (1,7.8) {0};
                \node[anchor=east,inner sep=1pt, blue] at (2,7.8) {9};
                \node[anchor=east,inner sep=1pt, blue] at (2.4,7.2) {4};
                \node[anchor=east,inner sep=1pt, blue] at (3.4,7.2) {1};

                \node[anchor=east,inner sep=1pt, blue] at (1,8.8) {1};
                \node[anchor=east,inner sep=1pt, blue] at (2,8.8) {4};
                \node[anchor=east,inner sep=1pt, blue] at (3,8.8) {7};
                \node[anchor=east,inner sep=1pt, blue] at (4,8.8) {2};

                \node[anchor=east,inner sep=1pt, blue] at (1,9.8) {0};
                \node[anchor=east,inner sep=1pt, blue] at (1.4,9.2) {3};
                \node[anchor=east,inner sep=1pt, blue] at (2.4,9.2) {6};
                \node[anchor=east,inner sep=1pt, blue] at (3.4,9.2) {3};

            \end{scope}
        \end{tikzpicture}
    }
    \subfloat{
        \begin{tikzpicture}[xscale=0.7,yscale=0.7]
            \initcube\printcube
            \flipup\printcube
            \flipup\printcube
            \flipup\printcube
            \flipup\printcube
            \flipup\printcube
            \rollup\printcube
            \rollup\printcube
            \rollup\printcube
            \rollup\printcube
            \flipright\printcube
            \rolldown\printcube
            \flipdown\printcube
            \rolldown\printcube
            \flipdown\printcube
            \flipdown\printcube
            \flipdown\printcube
            \flipdown\printcube
            \flipdown\printcube
            \flipdown\printcube
            \rollright\printcube
            \rollup\printcube
            \flipup\printcube
            \rollup\printcube
            \flipup\printcube
            \flipup\printcube
            \flipup\printcube
            \rollup\printcube
            \flipup\printcube
            \rollup\printcube
            \flipright\printcube
            \rolldown\printcube
            \flipdown\printcube
            \rolldown\printcube
            \flipdown\printcube
            \flipdown\printcube
            \rolldown\printcube
            \rolldown\printcube
            \rolldown\printcube
            \rolldown\printcube
        
            \begin{scope}[shift={(-.5,-.5)}]
                \draw[gray, very thin] (0, 0) grid (4, 10); 
                \draw[blue, very thick] (0, 0) rectangle (4, 10);
                \draw[blue, very thick] (2, 1) -- (2, 3);
                \draw[blue, very thick] (3, 2) -- (3, 4);   
                \draw[blue, very thick] (1, 6) -- (1, 8);   
                \draw[blue, very thick] (2, 7) -- (2, 9);   

                \node[anchor=east,inner sep=1pt, blue] at (1,0.8) {7};
                \node[anchor=east,inner sep=1pt, blue] at (2,0.8) {3};
                \node[anchor=east,inner sep=1pt, blue] at (3,0.8) {9};
                \node[anchor=east,inner sep=1pt, blue] at (4,0.8) {0};

                \node[anchor=east,inner sep=1pt, blue] at (1,1.8) {14};
                \node[anchor=east,inner sep=1pt, blue] at (1.4,1.2) {13};
                \node[anchor=east,inner sep=1pt, blue] at (3,1.8) {8};
                \node[anchor=east,inner sep=1pt, blue] at (4,1.8) {3};

                \node[anchor=east,inner sep=1pt, blue] at (1,2.8) {11};
                \node[anchor=east,inner sep=1pt, blue] at (2,2.8) {12};
                \node[anchor=east,inner sep=1pt, blue] at (3,2.8) {9};
                \node[anchor=east,inner sep=1pt, blue] at (3.4,2.2) {0};

                \node[anchor=east,inner sep=1pt, blue] at (1,3.8) {10};
                \node[anchor=east,inner sep=1pt, blue] at (1.4,3.2) {9};
                \node[anchor=east,inner sep=1pt, blue] at (2.4,3.2) {10};
                \node[anchor=east,inner sep=1pt, blue] at (4,3.8) {1};

                \node[anchor=east,inner sep=1pt, blue] at (1,4.8) {5};
                \node[anchor=east,inner sep=1pt, blue] at (2,4.8) {6};
                \node[anchor=east,inner sep=1pt, blue] at (3,4.8) {8};
                \node[anchor=east,inner sep=1pt, blue] at (4,4.8) {1};

                \node[anchor=east,inner sep=1pt, blue] at (1,5.8) {4};
                \node[anchor=east,inner sep=1pt, blue] at (1.4,5.2) {3};
                \node[anchor=east,inner sep=1pt, blue] at (2.4,5.2) {7};
                \node[anchor=east,inner sep=1pt, blue] at (4,5.8) {2};

                \node[anchor=east,inner sep=1pt, blue] at (1,6.8) {0};
                \node[anchor=east,inner sep=1pt, blue] at (2,6.8) {2};
                \node[anchor=east,inner sep=1pt, blue] at (3,6.8) {6};
                \node[anchor=east,inner sep=1pt, blue] at (4,6.8) {3};

                \node[anchor=east,inner sep=1pt, blue] at (1,7.8) {0};
                \node[anchor=east,inner sep=1pt, blue] at (2,7.8) {2};
                \node[anchor=east,inner sep=1pt, blue] at (3,7.8) {7};
                \node[anchor=east,inner sep=1pt, blue] at (4,7.8) {4};

                \node[anchor=east,inner sep=1pt, blue] at (1,8.8) {0};
                \node[anchor=east,inner sep=1pt, blue] at (1.4,8.2) {1};
                \node[anchor=east,inner sep=1pt, blue] at (2.4,8.2) {6};
                \node[anchor=east,inner sep=1pt, blue] at (3.4,8.2) {5};

                \node[anchor=east,inner sep=1pt, blue] at (1,9.8) {0};
                \node[anchor=east,inner sep=1pt, blue] at (1.4,9.2) {1};
                \node[anchor=east,inner sep=1pt, blue] at (2.4,9.2) {5};
                \node[anchor=east,inner sep=1pt, blue] at (4,9.8) {4};

            \end{scope}
        \end{tikzpicture}
    }
    \caption{Foldings of all minimal 4 by n polynominoes that fold with only vertical slits}
    \label{fig:4slitsminimal}
\end{figure}

%% file: example_5holes.tex
\begin{tikzpicture}[xscale=0.7,yscale=0.7]
    \begin{scope}[shift={(0,-1)}]
    \initcube\rollup\printcube
    \rollup\printcube
    \rollup\printcube
    \rollup\printcube
    \rollright\printcube
    \flipright\printcube
    \rolldown\printcube
    \rolldown\printcube
    \rollleft\printcube
    \rolldown\printcube
    \flipright\printcube
    \flipright\printcube
    \flipright\printcube
    \rollup\printcube
    \rollleft\printcube
    \rollup\printcube
    \rollup\printcube
    \flipright\printcube
    \rolldown\printcube
    \flipright\printcube
    \rollup\printcube
    \flipright\printcube
    \rolldown\printcube
    \rolldown\printcube
    \rollleft\printcube
    \rolldown\printcube
    \flipright\printcube
    \flipright\printcube
    \flipright\printcube
    \flipright\printcube
    \rollup\printcube
    \flipleft\printcube
    \rollleft\printcube
    \rollup\printcube
    \rollup\printcube
    \flipright\printcube
    \rollright\printcube
    \rolldown\printcube
    \end{scope}
    
    \begin{scope}[shift={(-.5,-.5)}] 
        \draw[gray, very thin] (0, 0) grid (10, 4);
        \draw[blue, very thick] (0, 0) rectangle (10, 4);
        \draw[blue, very thick] (1, 3) rectangle (2, 2);
        \draw[blue, very thick] (8, 3) rectangle (9, 2);
        \draw[blue, very thick] (2, 1) -- (4, 1);
        \draw[blue, very thick] (4, 2) -- (6, 2);
        \draw[blue, very thick] (6, 1) -- (8, 1);

        \node[anchor=east,inner sep=1pt, blue] at (1,0.8) {0};
        \node[anchor=east,inner sep=1pt, blue] at (2,0.8) {1};
        \node[anchor=east,inner sep=1pt, blue] at (3,0.8) {2};
        \node[anchor=east,inner sep=1pt, blue] at (4,0.8) {3};
        \node[anchor=east,inner sep=1pt, blue] at (5,0.8) {4};
        \node[anchor=east,inner sep=1pt, blue] at (6,0.8) {5};
        \node[anchor=east,inner sep=1pt, blue] at (7,0.8) {6};
        \node[anchor=east,inner sep=1pt, blue] at (8,0.8) {7};
        \node[anchor=east,inner sep=1pt, blue] at (9,0.8) {8};
        \node[anchor=east,inner sep=1pt, blue] at (10,0.8) {9};

        \node[anchor=east,inner sep=1pt, blue] at (1,1.2) {0};
        \node[anchor=east,inner sep=1pt, blue] at (2,1.2) {1};
        \node[anchor=east,inner sep=1pt, blue] at (3,1.2) {0};
        \node[anchor=east,inner sep=1pt, blue] at (4,1.2) {1};
        \node[anchor=east,inner sep=1pt, blue] at (5,1.2) {2};
        \node[anchor=east,inner sep=1pt, blue] at (6,1.2) {3};
        \node[anchor=east,inner sep=1pt, blue] at (7,1.2) {2};
        \node[anchor=east,inner sep=1pt, blue] at (8,1.2) {3};
        \node[anchor=east,inner sep=1pt, blue] at (9,1.2) {4};
        \node[anchor=east,inner sep=1pt, blue] at (10,1.2) {9};

        \node[anchor=east,inner sep=1pt, blue] at (1,2.8) {0};
        \node[anchor=east,inner sep=1pt, blue] at (3,2.8) {1};
        \node[anchor=east,inner sep=1pt, blue] at (4,2.8) {2};
        \node[anchor=east,inner sep=1pt, blue] at (5,2.8) {3};
        \node[anchor=east,inner sep=1pt, blue] at (6,2.8) {4};
        \node[anchor=east,inner sep=1pt, blue] at (7,2.8) {5};
        \node[anchor=east,inner sep=1pt, blue] at (8,2.8) {6};
        \node[anchor=east,inner sep=1pt, blue] at (10,2.8) {7};

        \node[anchor=east,inner sep=1pt, blue] at (1,3.8) {0};
        \node[anchor=east,inner sep=1pt, blue] at (2,3.8) {0};
        \node[anchor=east,inner sep=1pt, blue] at (3,3.8) {1};
        \node[anchor=east,inner sep=1pt, blue] at (4,3.8) {2};
        \node[anchor=east,inner sep=1pt, blue] at (5,3.8) {3};
        \node[anchor=east,inner sep=1pt, blue] at (6,3.8) {4};
        \node[anchor=east,inner sep=1pt, blue] at (7,3.8) {5};
        \node[anchor=east,inner sep=1pt, blue] at (8,3.8) {6};
        \node[anchor=east,inner sep=1pt, blue] at (9,3.8) {7};
        \node[anchor=east,inner sep=1pt, blue] at (10,3.8) {1};

    \end{scope}

\end{tikzpicture}

%% file: thm3_good_case.tex
\begin{tikzpicture}[scale=0.7]
                \fill[gray!10] (0, 0) rectangle (7.5, 4);
                \fill[white] (1, 2) rectangle (2, 3);

                \begin{scope}[shift={(0.5,3.5)}]
                \initcube\printcube
                \rollright\printcube
                \flipright\printcube
                \flipright\printcube
                \flipright\printcube
                \flipright\printcube
                \flipright\printcube
                \rolldown\printcube
                \rolldown\printcube
                \rollup
                \flipleft\printcube
                \flipleft\printcube
                \flipleft\printcube
                \flipleft\printcube
                \rolldown\printcube
                \flipright\printcube
                \flipleft\rollleft\printcube
                \rolldown\printcube
                \flipleft\printcube
                \rollup\printcube
                \rollup\printcube
                \end{scope}

                \foreach \i in {1,2,...,7}
                    \draw (\i,0)--(\i,4);
                \foreach \j in {1,2,3}
                    \draw (0,\j)--(7.5,\j);

                % \draw[red, thick] (1, 4) -- (1, 3);
                % \draw[red, thick] (2, 2) -- (2, 1);
                % \draw[black, thick] (1, 2) -- (1, 0);
                % \draw[red, thick] (2, 3) -- (7.5, 3);
                % \draw[red, thick] (0, 3) -- (1, 3);
                % \draw[red, thick] (2, 2) -- (7.5, 2);
                % \draw[red, thick] (0, 2) -- (1, 2);
                % \draw[black, thick] (2, 4) -- (2, 3);
                
                % \draw[black, thick] (3, 4) -- (3, 0);
                % \draw[black, thick] (4, 4) -- (4, 2);
                % \draw[black, thick] (5, 4) -- (5, 0);
                % \draw[black, thick] (6, 4) -- (6, 2);
                % \draw[black, thick] (7, 4) -- (7, 0);
                % \draw[red, thick] (0, 1) -- (7.5, 1);
                
                \draw[green, thick] (2, 1) -- (2, 0);
                \draw[green, thick] (4, 2) -- (4, 0);
                \draw[green, thick] (6, 2) -- (6, 0);
                
                \draw[blue, thick] (0, 0) -- (7.5, 0);
                \draw[blue, thick] (0, 0) -- (0, 4);
                \draw[blue, thick] (0, 4) -- (7.5, 4);
                
                \draw[blue, very thick] (1, 3) rectangle (2, 2);
                \draw[blue, very thick] (2, 1) -- (4, 1);
                \draw[blue, very thick] (4, 2) -- (6, 2);
                \draw[blue, very thick] (6, 1) -- (7.5, 1);

                % \path (0.5, 3.5) node {1}
                %       (1.5, 3.5) node {3}
                %       (0.5, 2.5) node {2}
                %       (0.5, 1.5) node {6}
                %       (0.5, 0.5) node {5}
                %       (1.5, 0.5) node {5}
                %       (6.5, 1.5) node {4}
                %       (1.5, 1.5) node {6}
                %       (3.5, 1.5) node {4}
                %       (2.5, 1.5) node {4};
                % \foreach \x in {2,..., 6} {
                %     \node at(\x+0.5, 3.5) {3};
                % }
                % \foreach \x in {2,..., 6} {
                %     \node at(\x+0.5, 2.5) {2};
                % }

                \path (2.5, 0.5) node {a}
                    (3.5, 0.5) node {a}
                    (4.5, 0.5) node {b}
                    (5.5, 0.5) node {b}
                    (6.5, 0.5) node {c}
                    (4.5, 1.5) node {b}
                    (5.5, 1.5) node {b};
                \path (2, -0.2) node {$*$}
                    (4, -0.2) node {$\star$};
                
            \end{tikzpicture}

%% file: main.bbl
\begin{thebibliography}{1}

\bibitem{aichholzer2020}
Oswin Aichholzer, Hugo~A. Akitaya, Kenneth~C. Cheung, Erik~D. Demaine, Martin~L. Demaine, S\'andor~P. Fekete, Linda Kleist, Irina Kostitsyna, Maarten L\"offler, Zuzana Mas\'arov\'a, Klara Mundilova, and Christiane Schmidt.
\newblock Folding polyominoes with holes into a cube.
\newblock {\em Comput. Geom.}, 93:Paper No. 101700, 14, 2021.

\bibitem{original}
Oswin Aichholzer, Michael Biro, Erik~D. Demaine, Martin~L. Demaine, David Eppstein, S\'andor~P. Fekete, Adam Hesterberg, Irina Kostitsyna, and Christiane Schmidt.
\newblock Folding polyominoes into (poly)cubes.
\newblock {\em Internat. J. Comput. Geom. Appl.}, 28(3):197--226, 2018.

\bibitem{puzzle}
N.~Beluhov.
\newblock Cube folding.
\newblock \url{https://nbpuzzles.wordpress.com/2014/06/08/cube-folding}.
\newblock Accessed: 2025-10-14.

\bibitem{geometricfolding}
Erik~D. Demaine and Joseph O'Rourke.
\newblock {\em Geometric folding algorithms}.
\newblock Cambridge University Press, Cambridge, 2007.
\newblock Linkages, origami, polyhedra.

\bibitem{origametry}
Thomas~C. Hull.
\newblock {\em Origametry---mathematical methods in paper folding}.
\newblock Cambridge University Press, Cambridge, 2021.

\bibitem{florian}
Christian~Lindorfer Oswin~Aichholzer, Florian~Lehner.
\newblock Folding polyominoes into cubes.
\newblock {\em J. Comput. Geom.}, 2025.
\newblock to appear.

\bibitem{ben-github}
B.~Shirley.
\newblock \url{https://github.com/ben-shirley/polyomino_algorithm}.
\newblock Accessed: 2025-10-14.

\end{thebibliography}
